\documentclass[final]{siamonline190516}

\usepackage[toc,page,titletoc]{appendix}

\usepackage[utf8]{inputenc} 
\usepackage[T1]{fontenc}    
\usepackage{hyperref}       
\usepackage{url}            
\usepackage{booktabs}       
\usepackage{amsfonts}       
\usepackage{nicefrac}       
\usepackage{microtype}      
\usepackage[caption=false]{subfig} 
\usepackage{capt-of}
\usepackage{colortbl}
\usepackage{enumitem}
\usepackage{dgleich-math}
\usepackage{dgleich-tensor}
\usepackage{subdepth}
\usepackage{tabularx}
\usepackage[export]{adjustbox}

\usepackage{xspace}

\newcommand{\ErdosRenyi}{Erd\H{o}s-R\'{e}nyi\xspace}
\newcommand{\LambdaTAME}{$\Lambda$-TAME\xspace}

\let\oldthebibliography=\thebibliography
\let\oldendthebibliography=\endthebibliography

\usepackage{natbib}
\renewcommand{\cite}{\citep}

\let\thebibliography=\oldthebibliography
\let\endthebibliography=\oldendthebibliography

\usepackage{framed}
\usepackage{wrapfig}
\definecolor{shadecolor}{gray}{0.9}
\newcounter{AsideNumber}
\newcommand*{\aside}[2]{\refstepcounter{AsideNumber}%
        \label{#1}%
        \begin{wrapfigure}{r}{0.4\linewidth}%
                \vspace*{-10pt}\begin{shaded*}%
                        \footnotesize%
                        \textsc{Aside~\theAsideNumber.~}\textit{#2}%
                \end{shaded*}\vspace*{-10pt}%
        \end{wrapfigure}%
}

\renewcommand{\ones}{\bm{\mathbbm{1}}}
\let\minimize\relax 
\let\maximize\relax 
\let\tvec\relax
\DeclareMathOperator*{\minimize}{\text{minimize}}
\DeclareMathOperator*{\maximize}{\text{maximize}}
\DeclareMathOperator{\tvec}{\text{vec}}

\newcommand{\dataeq}{\underset{\text{vec}}{\Leftrightarrow}}

\newcommand{\ileave}[3][1pt]{{#2}_{\mkern -6mu \raisebox{#1}{$_{\bm{\perp}}$}} \mkern -6mu {#3}}
\newcommand{\mSigma}{\mat{\Sigma}}
\newcommand{\kdel}{\cmE}

\usepackage{todonotes}
\usepackage{ellipsis} 
\newcommand{\lcdots}{{...}}
\usepackage{bbm}
\usepackage{tikz}
\usetikzlibrary{matrix,calc}

\makeatletter
\def\endthebibliography{%
	\def\@noitemerr{\@latex@warning{Empty `thebibliography' environment}}%
	\endlist
}
\makeatother

\newlength\mylinewidth
\tikzset{
	ultra thin/.style= {line width=0.25\mylinewidth},
	very thin/.style=  {line width=0.5\mylinewidth},
	thin/.style=       {line width=\mylinewidth},
	semithick/.style=  {line width=1.5\mylinewidth},
	thick/.style=      {line width=2\mylinewidth},
	very thick/.style= {line width=3\mylinewidth},
	ultra thick/.style={line width=4\mylinewidth},
	every picture/.style={very thick}
}

\definecolor{t4}{HTML}{fb6b5b} 
\definecolor{t4darker}{HTML}{E34635} 
\definecolor{t4darkest}{HTML}{BB2617} 

\definecolor{t5}{HTML}{fbc04c} 

\definecolor{t3}{HTML}{aa42a3} 

\definecolor{t2lightest}{HTML}{A7EBD8} 
\definecolor{t2}{HTML}{00a878} 
\definecolor{t2darkest}{HTML}{00664A} 

\definecolor{t1}{HTML}{7ecbe0} 
\definecolor{t1darker}{HTML}{56AFC9} 
\definecolor{t1darkest}{HTML}{3590AA} 

\definecolor{darkred}{RGB}{200,0,0}

\newcommand{\ib}{\ensuremath{\boldsymbol{\mathrm{i}}}}
\newcommand{\ibp}{\boldsymbol{\mathrm{i}'}}

\newcommand{\jb}{\boldsymbol{\mathrm{j}}}

\newcommand{\jbp}{\boldsymbol{\mathrm{j}'}}

\newcommand{\lb}{\ensuremath{\boldsymbol{\ell}}}
\newcommand{\lbp}{\ensuremath{\boldsymbol{\ell'}}}

\newcommand{\modetimes}[2][p]{\!\times_1\!{#2}\!\times_2\!\cdots\!\times_{#1}\!{#2}}

\hypersetup{  
	colorlinks = false,
	linkbordercolor = t4,
	citebordercolor = t2,
	filebordercolor = t1,
	urlbordercolor = t1,
	menubordercolor = t3,
}
\renewcommand*{\backref}[1]{}%
\renewcommand*{\backrefalt}[4]{%
  \ifcase #1 %
    No citations.%
  \or
    Cited on page #2.%
  \else
    Cited on pages #2.%
  \fi
}%

\usepackage{algorithm}
\usepackage[noend]{algpseudocode}
\usepackage{float}

\usepackage{lipsum}
\usepackage{amsfonts}
\usepackage{graphicx}
\usepackage{epstopdf}
\ifpdf
  \DeclareGraphicsExtensions{.eps,.pdf,.png,.jpg}
\else
  \DeclareGraphicsExtensions{.eps}
\fi

\usepackage{enumitem}
\setlist[enumerate]{leftmargin=.5in}
\setlist[itemize]{leftmargin=.5in}

\newsiamremark{remark}{Remark}
\newsiamremark{hypothesis}{Hypothesis}
\crefname{hypothesis}{Hypothesis}{Hypotheses}
\newsiamthm{claim}{Claim}

\headers{Dominant Z-Eigenpairs of Tensor Kronecker Products}{C. Colley,  H. Nassar, and D. Gleich}

\title{Dominant Z-Eigenpairs of Tensor Kronecker Products are Decoupled and Applications to Higher-Order Graph Matching\thanks{Submitted to the editors June 9, 2022.\funding{This research is supported in part by NSF IIS-1546488, CCF-1909528, the NSF Center for Science of Information STC, CCF-0939370, DOE award DE-SC0014543, and the Sloan Foundation.}}}

\author{Charles~Colley,~Huda Nassar, and David F.~Gleich\thanks{Colley and Gleich are at the Computer Science department at Purdue University. Nassar is now at RelationalAI and the research was completed at Purdue and Stanford.}}

\ifpdf
\hypersetup{
  pdftitle={An Example Article},
  pdfauthor={C. Colley, H. Nassar, and D. Gleich}
}
\fi

\begin{document}

\maketitle

\begin{abstract}
Tensor Kronecker products, the natural generalization of the matrix Kronecker product, are independently emerging in multiple research communities. Like their matrix counterpart, the tensor generalization gives structure for implicit multiplication and factorization theorems. We present a theorem that decouples the dominant eigenvectors of tensor Kronecker products, which is a rare generalization from matrix theory to tensor eigenvectors. This theorem implies low rank structure ought to be present in the iterates of tensor power methods on Kronecker products. We investigate low rank structure in the network alignment algorithm TAME, a power method heuristic. Using the low rank structure directly or via a new heuristic embedding approach, we produce new algorithms which are faster while improving or maintaining accuracy, and scale to problems that cannot be realistically handled with existing techniques. 
\end{abstract}



\section{Introduction}\label{sec:introduction}

		Given the ubiquity of the matrix Kronecker product, it is no surprise tensor analogs have been considered where 
		\[ \mbox{\centering\includegraphics[width=\textwidth]{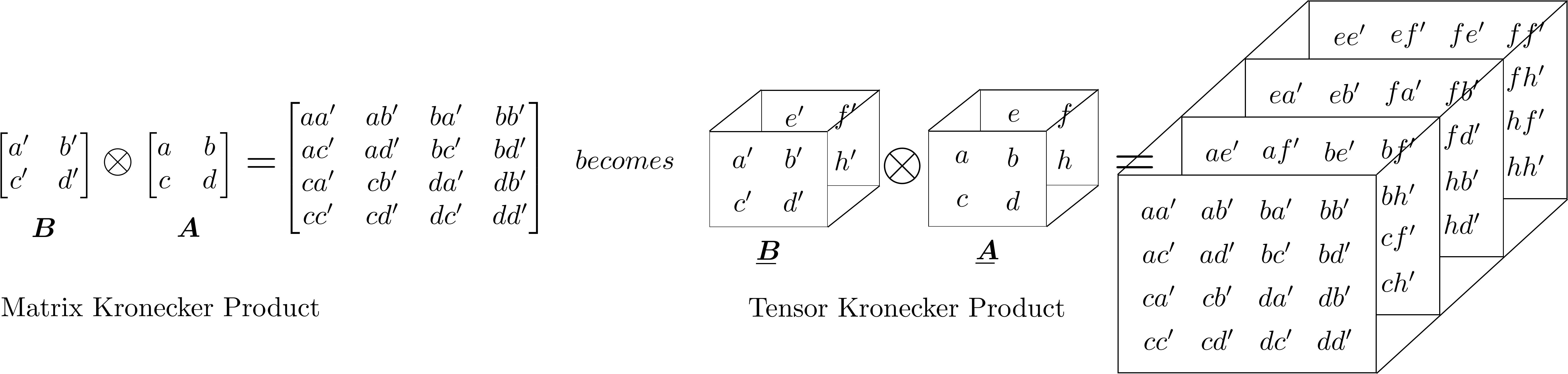}} \] 
		Existing research around this straightforward generalization of Kronecker products includes random graph models~\cite{akoglu2008rtm,Eikmeier-2018-hyperkron}, image and tensor completion~\cite{phan2012revealing,sun2018tensor}, generalized CP decompositions~\cite{batselier2017constructive,phan2013basis}, and graph alignment~\cite{park2013fast,mohammadi2017triangular,shen2018genome}. Generalizations of Kronecker product multiplication~\cite{sun2016moore,shao2013general}, folding~\cite[section 4.3.6]{ragnarsson2012structured}, and structure inheritance properties~\cite{batselier2017constructive} have been discussed as well. We introduce a new theorem (\S\ref{sec:extremal-eigenbound}) which shows that the dominant z-eigenvector of the tensor $ \cmB \kron \cmA $ decouples into the dominant eigenvectors of $ \cmB $ and $ \cmA $. This result is a simple generalization of the matrix case, which is surprising given the many differences between tensor Z-eigenvectors and matrix eigenvectors. 
		
		Differences with the matrix case persist, though. For matrices, we have the stronger result that all eigenvectors of $\mB \kron \mA$ (up to invariant subspaces) are Kronecker products of eigenvectors of $\mA$ and $\mB$. This is not true for even a diagonal tensor. 	
		Consider the example of the eigenvectors of a $ 4 \times 4 \times 4 $ diagonal tensor with ones on the diagonal $\cmD_4$.  This tensor can be decomposed into two diagonal tensors, $\cmD_4 = \cmD_2 \kron \cmD_2$. Using the software associated with \cite{cui2014all} (or see \cite[Ex. 2.2]{cartwright2013number}), the eigenvalues of $\cmD_2$ are $\pm  1 $ and $ \pm {1}/{\sqrt{2}} $, where the eigenvectors corresponding to $ 1 $ are columns of the $2 \times 2$ identity matrix $\mI$, and the eigenvector corresponding to $ {1}/{\sqrt{2}} $ is $ {1}/{\sqrt{2}} \bmat{ 1 & 1 } $. However the eigenvalues of $ \cmD_4 $ are $  \pm1 $, $  \pm{1}/{2} $, $  \pm1/{\sqrt{3}} $, and $  \pm{1}/{\sqrt{2}} $.  The eigenvector for $\pm 1/\sqrt{3}$ is any permutation of $(1/\sqrt{3}) \bmat{ 1 & 1 & 1 & 0}^{\smash{T}}$. None of these can be decomposed into the Kronecker product of any eigenvectors of $ \cmD_2 $. Our theorem simply says the dominant eigenvector has such a decomposition. The existence of these other eigenvectors is expected as the set of projectively equivalent eigenvectors of symmetric tensors is exponential in the dimension of the tensor \cite[Thm. 1.2]{cartwright2013number}.

	As an application of our new theorem, we discuss how we can apply both existing theory on mixed-products and our new theorem to facilitate the use of large motifs -- small repeating subgraphs -- in the network alignment algorithm TAME~\cite{mohammadi2017triangular}. 
	Network alignment problems date to the 1950's due to the relationship with the quadratic assignment problem and facility location~\cite{bazaraa1983branch,koopmans1957assignment,lawler1963quadratic}.
	Methods to address the problem range from relax and round on integer problems~\cite{Burkard-2012-assignment-problems,lawler1963quadratic}, to Lagrangian relaxation techniques~\cite{klau2009new}, to well-motivated heuristics~\cite{singh2008global}, and many other types of techniques~\cite{patro2012global,malod2015graal}. Recent eigenvector-inspired spectral approaches~\cite{singh2008global,kollias2011network,feizi2019spectral,nassar2018low} have been among the most scalable and versatile.  
	TAME is a graph matching algorithm that uses a tensor Kronecker product and builds upon frameworks which align edges~\cite{blondel2004measure,zass2008probabilistic} to align motifs such as triangles. In TAME, the key computation is a tensor power method on $ \cmB \kron \cmA $, which is only ever manipulated implicitly due to its size. Implicit manipulation addresses the memory complexity, but means the computational complexity is quadratic in the number of motifs of the network. Here, our theorem suggests that iterates of the power method should have low rank. We find this to be the case and design the LowRankTAME (see \S\ref{sec:lowrank-tame}) to use that structure as well as a new algorithm that enforces rank-1 structure \LambdaTAME (see \S\ref{LambdaTAME}) for additional scalability.  
	
	
	The new algorithm, LowRankTAME, gives around a 10-fold runtime improvement while producing the same iterates as TAME. 
	 The new algorithm \LambdaTAME uses the decoupling in Theorem~\ref{thm:spectrum-general} to independently processes graphs $\mA$ and $\mB$ akin to the NSD algorithm~\cite{kollias2011network}. When \LambdaTAME is combined with careful refinement involving nearest neighbor queries, local search, or Klau's algorithm~\cite{klau2009new}, it has faster end-to-end runtimes and better performance (the number of edges and motifs matched increases).  Moreover, it scales up to aligning $9$-cliques, which is well beyond the capabilities of existing algorithms. In fact, the discrete operations involving matching and refinement now dominate runtime compared with the linear algebra.

%
	The remainder of our paper formally establishes these results. It begins with a brief overview of our preliminaries for tensors (\S\ref{Preliminaries}). From there we move on to our three primary contributions:
	\begin{enumerate}
		\item[(1)] a novel extremal Z-eigenbound for tensor Kronecker products (\S\ref{sec:extremal-eigenbound})
		\item[(2)] an exploration of how newly found low rank structure can accelerate the graph matching algorithm TAME (\S\ref{sec:lowranktame}),
		\item[(3)] a new algorithm \LambdaTAME which outperforms TAME in speed and accuracy, and allows us to align larger motifs efficiently (\S\ref{sec:lambdatame}). 
	\end{enumerate}
	We evaluate our work by exploring larger motifs than previously considered feasible in \S\ref{sec:experiments}. We align protein protein interaction (PPI) networks from the Biogrid repository~\cite{stark2006biogrid} collected in the Local vs.~Global Network Alignment collection (LVGNA)~\cite{meng2016local} along with random geometric networks perturbed with partial duplication~\cite{bhan2002duplication,chung2003duplication,hermann2014large} and \ErdosRenyi noise models~\cite[section 3.4]{feizi2019spectral}. Our code is available (see \S\ref{sec:experiments}) and we have attempted to make our results as reproducible as possible by including the experiment driver codes as well.

\section{Definitions \& Preliminaries}\label{Preliminaries}

\subsection{General Matrix and Graph Notation} 
We use upper case bold letters for matrices,  $\mA$, $\mX$, and lower case bold letters for vectors $\vx, \vy$. We use colons over an index to denote a row or column of a matrix, akin to Matlab. The vector of all ones of length $n$ is $\ones_n$. A graph consists of a vertex set $V$ an edge set $E$. It can be weighted with a positive edge weight for each edge and an implicit zero weight for each non-edge, or it can be unweighted in which case edges have an implicit weight of 1 and non-edges have weight 0. All of the graphs we consider are undirected. The adjacency matrix then corresponds with a \emph{symmetric} matrix of edge weights for a fixed order of the vertices.  	 	
 
\subsection{Tensor notation and tensor eigenvectors}
\label{sec:tensor-evec} 
Tensors are denoted by bold underlined upper case letters, $\cmA$, $\cmT$, $\cmS$. We use a bold tuple of indices $\ib = (i_1,\dots,i_k)$ to denote each element of the tensor $\cA(\ib) = \cA(i_1,\dots,i_k)$ \cite[section 12.4.2]{Golub-2013-book}. The axes of a tensor are called modes. A matrix is a 2-mode tensor, and we consider $k$-modes in general. We call a tensor symmetric if entries are the same in all permutations of the tensor modes. When the dimension for each mode of a tensor is the same, we call this tensor cubical. 
We use the shorthand  $[n]^k$  to enumerate multi-indices $ \ib $ across all choices of  $1 \ldots n$ in each of the $k$ modes

Let $\cmA$ be a $k$-mode, cubical tensor of dimension $n$. We make frequent use of the polynomial  
\[ \textstyle \sum_{\ib \in [n]^k} \cA(\ib) x(i_1) x(i_2) \cdots x(i_k), \text{ written as } \cmA \vx^k, \] 
which generalizes the quadratic $\vx^T \mA \vx = \sum_{ij} A_{ij} x_i x_j$ (and we could write as $\mA \vx^2$). We adopt a functional notation in general.  When contracting $ p \leq k $ modes (with potentially $ p $ different vectors), we express
\[ 
  \textstyle\sum_{\lb \in [n]^p} \cA(\lb,\jb) x_1(\ell_1) \cdots x_p(\ell_p) \text{ as } \cmA(\vx_1, \ldots, \vx_p), 
\]
where $ \jb \in [n]^{k-p} $ indexes the trailing uncontracted modes. When contracting the same vector in each mode, we can simply write $ \cmA\vx^p $. 

We call multiplying by a matrix $ \mX \in \RR^{n \times r}$ in a given mode, a modal product. The modal product of the first mode produces a non-symmetric tensor, traditionally written as  $ \cmA \times_1 \mX $, whose first mode becomes dimension $ r $ and other modes remain dimension $ n $. Modal products may be extended to $ \cmA \modetimes{\mX} $ for $ p \leq k $ products.  For $ p $ matrices of dimension $ n $ by $ r $ matrices and the indices $  \ib \in [r]^{p}$ and $ \jb \in [n]^{p} $, we write a general modal product as 
\begin{equation}
	[\cmA \times_1 \mX_1 \times_2 \cdots \times_p \mX_p](\ib, \jb) 
	=  [\cmA\bigl(\mX_1(:,i_1), \mX_2(:,i_2), \ldots, \mX_p(:,i_p)\bigr)](\jb).
\end{equation} 

There are a variety of notions of tensor eigenvectors~\cite{qi2017tensor}. We use the $Z$-eigenvector of \citet{qi2005eigenvalues} or the $\ell_2$ eigenvectors of \citet{lim2005singular}. 
A Z-eigenpair of a tensor $ \cmA $, is a pair $ (\lambda,\vx)$ with $\lambda$ scalar and $\vx$ an $n$-vector, where 
\begin{equation} \label{eq:tensor-evec}
\cmA \vx^{k-1} \!=\! \lambda \vx \quad \normof[2]{\vx}\! = \!1.
\end{equation}
Equivalently, a tensor Z-eigenpair is a KKT point of the optimization problem 
 \begin{equation} \label{eq:tensor-evec-opt}
 	\maximize \quad {\cmA \vx^k} \quad \subjectto \quad {\normof[2]{\vx}^k = 1.}
\end{equation}
There is an exponentially increasing number of tensor eigenvectors as the number of modes $k$ grows~\cite{cartwright2013number}. For symmetric tensors, the eigenvectors can be computed via the Laserre hierarchy and convex programming~\cite{cui2014all} or Newton methods~\cite{jaffe2018newton}, although these techniques do not scale to large tensors. The higher-order power method (HOPM)~\cite{de1995higher}, symmetric shifted higher-order power method (SSHOPM)~\cite{kolda2011shifted}, its generalizations~\cite{kolda2014adaptive}, and dynamical systems~\cite{benson2019computing} are among the scalable ways to compute tensor eigenvectors. Although these scalable methods may not have the most satisfactory theoretical guarantees, they are practical and useful. 
	
\subsection{Kronecker products of tensors and vectorization}
\label{sec:kron-intro}

The Kronecker product $\kron$ between matrices arises from treating the pair of matrices $\mB$ and $\mA$ in $\mY = \mA \mX \mB^T$ as a linear operator from $\mX$ to $\mY$. (The order arises from how the matrix $\mX$ is linearized, see more below.) The way we present the definition of $\mB \kron \mA$ for \emph{matrices} involves a more complicated-seeming \emph{interleaving} of indices from $\mA$ and $\mB$, but this will enable a seamless generalization to tensors. Let $\ileave{i}{i'}$ represent a linearization, or vectorization, of the pair $i,i'$ to a single index. For instance, if both $i, i'$ range from 1 to $m$, then $\ileave{i}{i'}$ represents the linearized index $i + m(i'-1)$ where we have vectorized by the first index. This joint index notation is exactly the vectorization: 
\[ \tvec(\mX)[\ileave{i}{i'}] = X(i,i'),\]
where we use the ``matrix-to-vector'' operator $\tvec$, which converts matrix-data into a vector by columns.
We extend this definition to an interleaving of index pairs as in $\ileave{(i,j)}{(i',j')} \to (\ileave{i}{i'}, \ileave{j}{j'})$. This gives us the matrix Kronecker product
\[ (\mB \kron \mA)[\ileave{i}{i'}, \ileave{j}{j'}] = A(i,j) B(i',j'). \]
Using this notation we have for $\mY = \mA \mX \mB^T$, let $\vy = \tvec(\mY), \vx = \tvec(\mX)$ and 
\[  y[\ileave{i}{i'}] = \sum\limits_{\ileave{j}{j'}} (\mB \kron \mA)[\ileave{i}{i'}, \ileave{j}{j'}] x[\ileave{j}{j'}] =  \sum\limits_{\ileave{j}{j'}} A(i,j) B(i',j') x[\ileave{j}{j'}]. \] 


The nice thing about this notation is it gives us a seamless way to generalize to tensors. Given two $k$-tuples of indices $\ib$ and $\ibp$ we have $\ileave{\ib}{\ibp} = (\ileave{i_1}{i_1'}, \ldots, \ileave{i_k}{i_k'})$. Then if $\cmA$ and $\cmB$ are two $k$-mode cubical tensors, we have the element-wise definition
\[ (\cmB \kron \cmA)[\ileave{\ib}{\ibp}] = \cA(\ib) \cB(\ibp). \] 
Equivalently, we can define this in terms of single element tensors. Let $\kdel_{\ib}$ be a tensor with a 1 in the $\ib$ entry and zero elsewhere. Then $\cmA = \sum_{\ib} \cA(\ib) \kdel_{\ib}$ and $\cmB = \sum_{\ibp} \cB(\ibp) \kdel_{\ibp}$, and 
\[ \cmB \kron \cmA  =  ( \sum\limits_{\ibp} \cB(\ibp) \kdel_{\ibp}\!) \kron (\textstyle \sum\limits_{\ib\vphantom{'}} \cA(\ib) \kdel_{\ib}) =  \sum\limits_{\mathclap{\ib,\ibp}} \cA(\ib) \cB(\ibp)  \kdel_{\ibp} \kron \kdel_{\ib\vphantom{'}}. \] 
Using the element-wise definition above $\kdel_{\ibp} \kron \kdel_{\ib\vphantom{'}}$ only has a single non-zero in the $\ileave{\ib}{\ibp}$ entry. 


\section{The dominant eigenvector of the multi-linear Kronecker product} 
	\label{sec:Kronecker}
	
 Our primary contribution within  this section is the eigenvalue theorem (Theorem~\ref{thm:spectrum}) which establishes a relationship between the dominant eigenpairs of the operands in $ \cmB \otimes \cmA $ and the dominant eigenpair of that tensor. 
	Both the proof of our main theorem as well as the faster graph matching computations 	
use a number of results about computing the contraction
\begin{equation} 
(\cmB \kron \cmA) \tvec(\mX)^{k-1}	\text{ when }\mX\text{ is low-rank. }
\end{equation}
The contraction lemma's we use are known~\cite{ragnarsson2012structured,shao2013general,sun2016moore,batselier2017constructive}. We include our own proofs in the supplement~\ref{sup:proofs} for completeness and, if needed, to build intuition about our specific notation.


\subsection{Existing contraction lemmas in our notation}
	We begin with a generalizations of two core matrix Kronecker contraction theorems that we make use of in our proof and application: $(\mB \kron \mA)(\vy\kron\vx) = (\mB \vy)\kron (\mA \vx) $ and $ (\mB \kron \mA)\tvec(\mX) = \tvec(\mA\mX\mB^T)$. Each lemma will allow us to compute the contractions in terms of the rank 1 components of $ \mX $. The lemmas are critical to power method algorithms because forming $\mB \kron \mA$ is prohibitively expensive even in the matrix case. When the rank of $ \mX $ and the orders of the tensors are small this is a  more effective strategy than implicit contraction with the dense form of $ \mX $. 

	\begin{lemma}\label{lmm:rank1}
		Given two $k$-mode,  cubical tensors $ \cmA $ and $ \cmB$ of dimension $m$ and $n$, respectively, and the $m \times n$ rank 1 matrix $\mX = \vu \vv^T$,  then for $1 \le p \le k$,
		\begin{align}
			(\cmB \otimes \cmA)\tvec(\mX)^{p} =  (\cmB \otimes \cmA)(\vv \otimes \vu)^{p} = \cmB\vv^{p} \otimes \cmA\vu^{p}.\label{eq:rank1MCP}
		\end{align} 
	\end{lemma}
	 The proof of which can be found using unfolding theorems~\cite[TKP Property 4]{ragnarsson2012structured}, or a matrix specific version can be found in \citet[eq. 3.3]{batselier2017constructive}.  Lemma~\ref{lmm:rank1} allows us to completely decouple the contractions between the two operands. The mixed product property actually generalizes in its entirety and can be found in~\cite[Thm. 3.1]{shao2013general} \textit{\&}~\cite[Prop. 2.3]{sun2016moore}. 
    The second lemma allows us to work with a general matrix $\mX = \mY \mZ^T$ where $\mY$ and $\mZ$ have $r$ columns. For the matrix case, we have $\tvec(\mA \mX \mB^T) = \tvec(\mA \mY \mZ^T \mB^T) = (\mB \mZ \kron \mA \mY) \tvec(\mI)$, where $\mI$ is the $r \times r$ identity. 
\begin{lemma}\label{lmm:rank-n}
Given two $k$-mode,  cubical tensors $ \cmA $ and $ \cmB$ of dimension $m$ and $n$, respectively, and the matrix $\mX \in \RR^{m \times n}$ of rank $r$ with the $r$ column decomposition $ \mX = \mY \mZ^T$,  then
\begin{equation*} \begin{aligned}
  (\cmB \otimes \cmA)\tvec(\mX)^{p}  &=  \textstyle\sum\limits_{\mathclap{\ib = [r]^p}} \cmB(\mZ(:,i_1), \ldots, \mZ(:,i_p)) \kron \cmA(\mY(:,i_1), \ldots, \mY(:,i_p)) \\
 &= ((\cmB \modetimes{\mZ}) \kron (\cmA \modetimes{\mY}))\tvec(\mI)^p,
\end{aligned} \end{equation*}
where $ \mI $ is the $ r \times r $ identity matrix. 
	\end{lemma}
    The proof can be found in \cite[TKP Property 3]{ragnarsson2012structured}. Lemma~\ref{lmm:rank-n} is what we use to compute contractions when the rank of $ \mX $ or the order of the motifs are sufficiently small enough. We include self contained proofs of each lemma using our notation in sections~\ref{sec:lmm-rank1} \textit{\&}~\ref{sec:lmm-rank-n}. 

		\subsection{Dominant Z-eigenpairs}
		\label{sec:extremal-eigenbound}
		A useful property of Kronecker products is that the eigenvalues and eigenvectors of $\mB \kron \mA$ decouple into Kronecker products of the eigenvectors of $\mA$ and $\mB$, individually. This makes spectral analysis of matrix Kronecker products efficient. 
		We call an eigenpair \emph{dominant} if it is the global maximum of $|\cmA \vx^k|$ where $\normof{\vx} = 1$. Here, we show that this decoupling property remains true for the dominant tensor eigenvector of a Kronecker product of tensors. 
		
		\begin{theorem} \label{thm:spectrum}\label{thm:spectrum-general}
		Let $\cmA$ be a symmetric, $k$-mode, $m$-dimensional tensor and $\cmB$ be a symmetric, $k$-mode, $n$-dimensional tensor. Suppose that $(\lambda_A^*, \vu^*)$ and $(\lambda_B^*, \vv^*)$ are any dominant tensor Z-eigenvalues and vectors of $\cmA$ and $\cmB$, respectively. Then $(\lambda_A^*\lambda_B^*, \vv^* \kron \vu^*)$ is a dominant eigenpair of $\cmB \kron \cmA$.  Moreover, any Kronecker product of Z-eigenvectors of $\cmA$ and $\cmB$ is a Z-eigenvector of $\cmB \kron \cmA$. 
		\end{theorem}
		\begin{proof}
		Let $\vx = \tvec(\mX)$ be any vector with $\normof[2]{\vx} = \normof[F]{\mX} = 1$ where $\mX$ is an $m \times n$ matrix.  
		Let $z(i) = \normof[2]{\mX(:,i)}$. We have $\normof[2]{\vz} = 1$ as well.  Then we create 
		\[ 	
		\mZ	 = \text{diag}(z(1), \ldots, z(n)) 
		\text{ and $\mY$ so that $\mX = \mY \mZ$ }.
		\]
		The $i$th column of $\mY$ is either normalized or entirely 0 (if $z(i) = 0$). 
		Recall that the dominant eigenpair maximizes $|(\cmB \kron \cmA) \vx^k| = |(\cmB \kron \cmA) \tvec(\mX)^{k}|$. From Lemma~\ref{lmm:rank-n} we have 
			\[ \begin{aligned}
		  |(\cmB \kron \cmA) \tvec(\mX)^{k}| 
		 & \textstyle = | \sum\limits_{\ib}  \cmA(\mY(:,i_1), \lcdots ,\mY(:,i_k)) \cmB(\mZ(:,i_1), \lcdots, \mZ(:,i_k))| \\
		& \textstyle =  | \sum\limits_{\ib}   \cmA(\mY(:,i_1),\lcdots,\mY(:,i_k)) \prod\limits_{j=1}^k z(i_j) \cmB(\mI(:,i_1),\lcdots,\mI(:,i_k))|,\\
		 \end{aligned} \] 
		 where $\mI(:,j)$ is the $j$th column of the identity matrix. 
		 Now, because $\cmA$ is symmetric, we have that  
		 \[ |\lambda_A^*| = \maximize \; |\cmA(\vu_1, \ldots, \vu_k)| \; \subjectto \; \normof{\vu_i} = \{0,1\}. \]
		 This follows from a result on the best rank-1 approximation of a symmetric tensor~\cite[Theorem 2.1]{zhang2012best}, where the result is with $\normof{\vu_i} = 1$. We can handle cases where $\normof{\vu_i} = 0$ (which we could have for zero columns of $\mX$) by simply noting that such an $ \mX $ would make $ (\cmB \kron \cmA) \tvec(\mX)^{k}  =0$, so the maximum will never occur for those. Thus, this gives us an upper-bound on $|\cmA(\mY(:,i_1), \lcdots ,\mY(:,i_k))|$ 
		\[ 
		\begin{aligned}
		& |(\cmB \kron \cmA) \tvec(\mX)^{k}| \le	 
		 |\lambda_A^*| \cdot | \sum\limits_{\ib}  \prod\limits_{j=1}^k z(i_j) \cmB(\mI(:,i_1),\lcdots,\mI(:,i_k)) |.
		\end{aligned}
		\] 
		Here, $\cmB(\mI(:,i_1),\lcdots,\mI(:,i_k))$ is just $\cB(\ib)$ and 
		$ |\sum\limits_{\ib}  \prod\limits_{j=1}^k z(i_j) \cB(\ib)| = |\cmB \vz^k| \le |\lambda_B^*|.$
		Putting the pieces together, we have that the dominant Z-eigenvalue of $\cmB \kron \cmA \le |\lambda_A \lambda_B|$.
		
		Now we show that Kronecker products of eigenvectors are also eigenvectors. Let $\vu$ and $\vv$ be any Z-eigenvectors of $\cmA$ and $\cmB$, with eigenvalues $\lambda_A$ and $\lambda_B$ respectively, then  
		\[\begin{aligned}  
		(\cmB \kron \cmA)(\vv \kron \vu)^{k-1} & = (\cmB \vu^{k-1}) \kron (\cmA\vv^{k-1}) = \lambda_B \vu \kron \lambda_A \vv 
		 = (\lambda_A \lambda_B) (\vv \kron \vu). 
		\end{aligned}\] 
		Using $\vu^*$ and $\vv^*$ gives us an eigenvector that achieves the upper-bound $|\lambda^*_A \lambda_B^*|$. 
		 \end{proof}		 
		 
\textbf{Observations.} For  non-negative tensors $\cmA$ and $\cmB$ in Theorem~\ref{thm:spectrum-general}, then  $\lambda_A, \lambda_B$ are nonnegative because the components of the best rank 1 approximation are non-negative \cite[Prop.~6]{qi2018very}. Consequently, the dominant eigenvalue and eigenvectors of the Kronecker product are nonnegative. 


		We provide additional MATLAB routines (and precomputed results) making use of \cite{cui2014all,jaffe2018newton} as computational verification for theorem~\ref{thm:spectrum} with randomized symmetric tensors which are small enough to enumerate the entire spectrum of $ \cmA $ and $ \cmB $. As expected, we found no counter examples with random dense symmetric tensors generated with the tensor toolbox~\cite{bader2008efficient} where we sample the spectra with Cui et al.'s constrained polynomial optimization~\cite{cui2014all} and Jaffe et al's Newton correction method (and it's orthogonal variant)~\cite{jaffe2018newton}. We report the relative difference between largest magnitude z-eigenvalue found for $ \cmA $, $ \cmB $, and $ \cmB \kron \cmA $, and the inner product of the associated eigenvectors in figure~\ref{fig:DomTenEig_exps}. This identified no exceptions to our theorem up to computational tolerances.

		\begin{figure*}[t!]
			\centering
			\includegraphics[scale=.85]{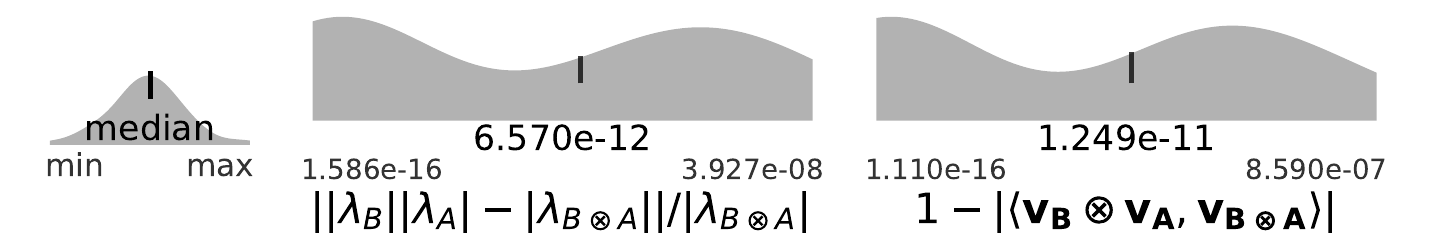}
			\vspace{-2.5mm}
			\caption{\label{fig:DomTenEig_exps}
			We computationally verify Theorem~\ref{thm:spectrum} in a handful of random problems. Our synthetic problems are of size $ m,n \in \{2,3,4\} $ and $ k \in \{3,4,5\} $. There are 30 trials overall and we show a density plot, median, min and max over the results. We report the differences of the dominant z-eigenpairs of tensor Kronecker product and its operands. For each tensor we use the largest magnitude $ \lambda $ found from each of the methods~\cite{cui2014all,jaffe2018newton} and its associated eigenvector. To measure eigenvector similarity, we use $ 1- |\langle \cdot,\cdot \rangle|$  where the absolute value addresses sign discrepancies. We initialize the NCM methods with 5000 uniformly drawn points from the unit sphere, and use default parameters Cui et al's methods. The NCM methods use a tolerance of $ 10^{-10} $ to measure differences in eigenvalues whereas Cui et al's methods use $ 10^{-4} $. We stored the tensors and results for future study by others in our codes. }
		\end{figure*}
						
%

\section{Faster Higher Order Graph Alignment Methods via Kronecker Structure}
	 \label{HOAlgos}
	 
	 We show how the tensor Kronecker product results from the previous section allow us to improve the higher-order network alignment algorithm TAME~\cite{mohammadi2017triangular} in two ways.  First, Lemma~\ref{lmm:rank-n} is the key to understanding how to use the Kronecker structure to make the iteration $ (\cmT_B \kron \cmT_A)\vx^{k-1}$ from TAME faster when $\vx = \tvec(\mX)$ and $\mX$ is low-rank and the number of modes -- equivalent to the size of the motif -- is not too large. 
	 Second, when expanding to larger motifs, such as 9-cliques, we will need to rely on our new simpler algorithm \LambdaTAME (\S\ref{sec:lambda-tame}), which is built upon our novel decoupling result, Theorem~\ref{thm:spectrum-general}. 

	\subsection{Background on higher-order graph alignment} 	

	Higher-order graph alignment considers motifs, or small subgraphs, beyond edges~\cite{conte2004-graph-matching,Bayati-2013-netalign,singh2008global} that are matched between a pair of networks. A motif is simply a graph -- usually small, like a triangle -- and an instance of a motif in a graph is simply an instance of an isomorphic induced subgraph~\cite{milo2002network}. From any graph we can induce a $ k $-regular hypergraph $H$ by identifying hyperedges with the presence of motifs with $ k $ vertices~\cite{estrada2006subgraph,klymko2014using,Benson-2015-tensor,mohammadi2017triangular}. Then the full edge set of the hypergraph involves enumerating all the instances of the motif. This can be computationally demanding to enumerate complicated motifs, but is fast for simple motifs like triangles and small cliques, and random sampling can make the process reasonable for larger cliques~\cite{jain2020power}. Analogously to the adjacency matrix, we use an adjacency tensor $\cmA$ to denote the presence of these motifs, or equivalently, hyperedges. 
	Formally, 
	\[
	\cA(i_1, \ldots, i_k) = 
	\begin{cases}
		1 \quad &\text{if} \,\, \text{nodes } {i_1}, \dots, {i_k} \text{ form motif } M, \\
		0 &\text{else.}
	\end{cases}
	\]
	Each permutation of the indices corresponds to a different orientation of the motif $ M $.
	Consequently, the adjacency tensor of a hypergraph is a symmetric, cubical tensor for the motifs we consider. 
	
	The higher-order graph alignment problem we consider is defined in terms of a matching between the vertices of two graphs~\cite{chertok2010efficient,park2013fast,mohammadi2017triangular}. For a pair of graphs $ A $ and $ B $ we characterize a matching between their vertex sets as a matrix.
	\begin{definition}[Matching Matrix]\label{MatchingMatrix}
		Let $ A $ and $ B $ be two graphs of size $ m $ and $ n $ respectively, then we define the matching matrix $ \mX \in \{0,1\}^{m \times n}$ such that
		\[
		\mX \ones_n \le  \ones_m, \quad\mX^T \ones_m  \leq \ones_n,\quad
		\text{and }\mX(i,i') =   
		\begin{cases}
			1 \quad\text{if $ i  \in V_A$ is matched to $ i' \in V_B$} \\
			0\quad \text{else}.
		\end{cases} 
		\]
	\end{definition}	
	
	We suppose we are given a similarity tensor $\cmS$ where entries can be indexed using a pair of tuples $\ib$ to represent the vertices of a motif in graph $A$ and $\ib'$ to represent the vertices of motif from graph $B$. The value $\cS(\ib, \ib')$ indicates the similarity of the motif at indices $\ib$ in graph $A$ to the motif at indices $\ib'$ in graph $B$. 
	A simple form of higher-order graph alignment problem is to optimize 
	\[ 
	\MAXone{}{\sum_{\ib} \sum_{\ib'} [\cS(\ib,\ib') X(i_1^{},i_1') X(i_2^{},i_2') \cdots X(i_k^{},i_k')]}{\mX \text{ is a matching}.}
	\] 
	The goal here is to find high-similarity entries $\cS(\ib,\ib')$ where the vertices involved in the motifs are matched.	This subsumes an edge-based alignment framework (such as~\citet{feizi2019spectral}) because $\ib$ could have just been the pair $(i,j)$. 
	
	We often find it convenient to write this objective as 
	\[ \MAXone{}{ \cmV \tvec(\mX)^k = \cmV \vx^k}{\mX \text{ is a matching}} \] 
	where we convert the similarity tensor $\cmS$ into a tensor $\cmV$ indexed with the same order with the $\tvec$ operator. This tensor to ``\emph{operator} for $\tvec(\mX)$'' transformation is something we repeatedly use and write it as 
	\begin{equation}
		\cmS \dataeq \cmV \text{ means } S(i_1,\cdots,i_k,i'_1,\cdots,i_k') = \cS(\ib,\ibp)= \cV[\ileave{\ib}{\ibp}] =V[\ileave{i_1^{}}{i'_1},\cdots,\ileave{i_k^{}}{i'_k}].
	\end{equation}
	This vec-form makes the \emph{eigenvector}-heuristic inspiration clear because eigenvectors optimize the generalized Rayleigh quotient $ \cmA \vx^k$. 
	
	 Many choices for $\cmS$ give rise to tensors $\cmV$ with Kronecker structure.  In TAME~\cite{mohammadi2017triangular}, we set $ \cmS $ to $1$ if there is a triangle at both $\ib$ in $A$ and $\ib'$ in $B$. 
	This idea gives Kronecker structure in $\cmV$. If we denoted the triangles of $A$ and $B$ in the triangle adjacency tensors $\cmT_A$ and $\cmT_B$ respectively, TAME's similarity tensor $\cmS$ would be $\cS(i,j,k,i',j',k') = \cT_A(i,j,k) \cT_B(i',j',k')$. 
	Though simple, this form is informative when given a matching $ \mX $, as 
 	\[
 \sum_{i\vphantom{'},j\vphantom{'},k\vphantom{'}} \sum_{i',j',k'} \underbrace{\cT_A(i,j,k)\cT_B(i',j',k')}_{\cS(\ib,\ib') = \cS(i,j,k,i',j',k')}  X(i,i') X(j,j') X(k,k') =
	 	6\bigg(\!\!
		 \begin{array}{c}
		    \text{the number of triangles}\\
		 	\text{aligned between A and B}   
		 \end{array}
	\!\!\bigg).
	\]
%

	This also gives us a  tensor $ \cmV = \cmT_B \kron \cmT_A $. Again, this framework is highly flexible. For example, the computer vision algorithm HOFASM~\cite{park2013fast} approximates a similarity tensor $ \cmS $ between pairs of triplets of a image features with a tensor of the form $ \cmV = \sum_{r,s} \cmB_{r,s} \kron \cmH_{r,s}$. For simplicity, we define the following objective function that will guide our subsequent research. 
	\begin{definition}[Global Graph Alignment]\label{GlobalGraphAlign}
		Fix graphs $ A $ and $ B $ to have $m $ and $ n $ vertices respectively, an $m \times n$ prior weight matrix $ \mW $, and a motif $ M $ with $ k $ vertices. Let $\cmS$ be a $2k$-mode similarity tensor where the $\cS(\ib,\ibp)$  entry denotes the similarity between the motifs induced by the vertices $\ib$ in graph $A$ and $ \ibp$ in graph $ B $. Then we wish to find a matching $ \mX$ between the vertices in $ A $ to the vertices in $ B $ which optimizes
		\[ 
		\MAXone{}{\sum_{\ib} \sum_{\ib'} [\cS(\ib,\ib') X(i_1^{},i_1') X(i_2^{},i_2') \!\cdots\! X(i_k^{},i_k')] + \sum_{i,i'} W(i,i') X(i,i')}{\mX \text{ is a matching}.}
		\] 
		Equivalently, we let $\cmV$ be the $k$-mode ``vec-operator'' form of $\cmS$, i.e.~$\cmS \dataeq \cmV$ or $\cV[\ileave{i_1}{i_1'},\ldots,\ileave{i_k}{i_k'}] \\= S(\ib,\ibp)$. Then the problem is
		\begin{equation} \label{eq:align-V}
		\MAXone{}{\cmV \tvec(\mX)^k + \trace(\mW^T \mX)}{\mX \text{ is a matching}}
		\end{equation}
		which makes the tensor-eigenvector inspiration clear (see Sec.~\ref{sec:tensor-evec}). 
	\end{definition}
	The tensor $\cmS$ will change depending on what structure we will consider for the higher order matching problem and we may adjust the weightings between the prior matrix and the affinity tensors (a similar edge based framework can be found~\cite{berg2005shape}). Note that $\cmS $ is not required to be symmetric in permutations of the first $k$ entries, which are permutations of $\ib$, but in the problems we consider in this paper, it will be. Likewise for permutations of the last $k$ entries for $\ibp$. Note that this means that $\cmV$ is a symmetric tensor, although $\cmS$ is not even cubical. 

	\subsection{TAME and LowRankTAME}
	\label{sec:lowranktame}
	TAME is a spectral method that uses a tensor-eigen\-vector heuristic to guide an alignment. It arises from the network alignment literature in bioinformatics.  
	The TAME method is a simple instance of the higher-order graph alignment framework (Definition~\ref{GlobalGraphAlign}) where,  given two graphs $A$ and $B$, we first enumerate triangles (or any motif of interest) in each, to build triangle adjacency tensors $\cmT_A$ and $\cmT_B$. Then we set $\cS(\ib,\ib') = \cT_A(\ib) \cT_B(\ib')$. For this choice, we have 
	\begin{equation}
	\cmS \dataeq \cmV = \cmT_B \kron \cmT_A.
	\end{equation} 
	This results in the following idealized optimization problem for TAME
    \begin{equation}\label{TAMEObjective}
	    \MAXone{}{\!\!(1\!-\!\alpha)\text{trace}(\mW^T\mX)+ \frac{\alpha}{6}(\cmT_B\kron\cmT_A)\tvec(\mX)^3}{\!\!\mX \text{ is a matching.}}
    \end{equation}
    Here the value $\alpha/6$ arises because each triangle alignment gives 6 entries in $\cmT_B \kron \cmT_A$ due to symmetry. 
	The weight matrix $ \mW $ gives flexibility to bias the alignment towards certain nodes. When no prior matrix is available, we make use of a rank 1 matrix $ \mW = \frac{1}{mn} \ones_m\ones_n^T $, which gives a uniform bias everywhere. 
	
	The heuristic procedure used in TAME is to deploy the SS-HOPM algorithm~\cite{kolda2011shifted} to seek a tensor eigenvector, or near tensor eigenvector, of $\mV = \cmT_B \kron \cmT_A$. We show the procedure in Algorithm~\ref{alg:tame}. 
	We present an affine-shift variant of the TAME method that includes the mixing parameter $\alpha$ to re-mix in the original iterate whereas TAME~\cite{mohammadi2017triangular} fixed $\alpha=1$. This choice sometimes helps boost performance a little bit. 

\textbf{Rounding with matching and scoring.} At each iteration, we explicitly \emph{round} the continuous valued $\mX$ and compute a matching using a max-weight matching algorithm. Then the procedure returns the best iterate with the highest downstream objective (triangle alignment, mixture, or some other combination). Returning the full iterate information is helpful for further refinement of the solution using a local search strategy described in \S\ref{sec:refinement}. This max-weight matching step, which is executed at each iteration, becomes expensive after we optimize the linear algebra using the Kronecker theory.


\begin{algorithm}[t]
\renewcommand{\algorithmiccomment}[1]{$\triangleright$\textit{#1}}
\caption{TAME~\cite{mohammadi2017triangular} with affine shift}
\label{OrigTAME} \label{alg:tame}
\begin{algorithmic}[1] 
	\Require $k$-mode  motif tensors $\cmT_A, \cmT_B$ for graphs $A$ and $B$, mixing parameter $ \alpha $, shift $ \beta $, tolerance $\eps$, weights $\mW$

	\Ensure Alignment heuristic $\mX$ and max-weight matching of $\mX$
    \State $ \mX_0 = \mW / \normof[F]{\mW}$ \hspace{33mm}\Comment{ Normalize first iterate} 
 	\For{$\ell = 0,1, \ldots$ until $|\lambda_{\ell+1} - \lambda_\ell| < \eps$}
 	\State \Comment{SS-HOPM iteration}
 	\State \hspace{\algorithmicindent} 
 			$\mX_{\ell+1} = \text{unvec}((\cmT_B \kron \cmT_A)\tvec(\mX_{\ell})^{k-1})$ \hspace{\algorithmicindent}\Comment{Implicitly}
   \State \hspace{\algorithmicindent} $ \lambda_{\ell+1} = \trace(\mX_\ell^T\mX_{\ell+1}) $ \hspace{12mm}\Comment{ Estimate tensor-eval }
   \State \hspace{\algorithmicindent} $ \mX_{\ell+1} \leftarrow \alpha \mX_{\ell+1} + \alpha \beta \mX_{\ell} + (1-\alpha) \mX_0 $
   \State \hspace{\algorithmicindent} $ \mX_{\ell+1} \leftarrow \mX_{\ell+1} / \normof[F]{\mX_{\ell+1}}$
   \State Set $t_{\ell+1}$ to be the score of a matching from $\mX_{\ell+1}$, e.g., number of motifs aligned
	\EndFor
  \State \Return $\mX_{\ell}$ and the matching of $\mX_{\ell}$ with the highest $t_{\ell}$
\end{algorithmic}
\end{algorithm}

	
	\textbf{Implicit multiplication.} In TAME the authors make use of an implicit operation to compute the iterates of the tensor powers 
	\begin{equation} \label{eq:tame-main}
		\cmT_B \kron \cmT_A \tvec(\mX)^{k-1}
	\end{equation}
	without forming $\cmT_B \kron \cmT_A$. This computation still takes $ O(\text{nnz}(\cmT_B)\text{nnz}(\cmT_A))$ work, where nnz is the number of non-zero entries in the sparse tensor. In the case of the uniform bias prior ($\mW = \frac{1}{mn}\ones_m \ones_n^T$), the first iterate is rank-1, so we could apply lemma~\ref{lmm:rank1} to decouple the operation.  Because of the shift $ \beta $, however, subsequent iterations will not remain rank 1 as the following observation clarifies.
	
	\textbf{Our observation.} Suppose that $\mW$ is rank 1 and we are dealing with a $k$-mode tensor. Then lemma~\ref{lmm:rank-n} applied to the TAME iteration, states that if $\mX$ is rank $r$ then the next iterate has rank at most $r^{k-1} + r + 1$. This follows from the number of combinations of vectors in the lemma combined with the addition of the $r$ rank factors for the previous iterate in the shift. Also $ r^{k-1} $ can be reduced to $ \smash{\binom{r + k - 2}{k-1}} $ for the symmetric case, but for simplicity we use the upper-bound $ r^{k-1} $.

	While this observation explains a simplistic analysis for the worst case scenario for the rank growth of the iterates, in practice we find it extremely conservative.  (See evidence in \S\ref{sec:lowrank-tame}.) This means that there is a useful low-rank strategy to employ with our theory. Namely, use Lemma~\ref{lmm:rank-n} to compute the components of the next iterate and then compute an exact low-rank factorization. As long as the rank does not get too big, this will be faster.
	
\textbf{An exact low-rank TAME iteration.} Let $\mW = \mF \mG^T$ be the low-rank factors of the weight matrix $\mW$ and let $t$ be the rank of the initial matrix.	The key idea of low-rank TAME, is to compute a rank $ r $ factorization of the iterate $ \mX_{\ell} $ and use lemma~\ref{lmm:rank-n} to compute all the $ r^{k-1} $ terms in the summation expansion to give us $\mX_{\ell+1} = \mU_{\ell+1} \mV_{\ell+1}^T$. (This is $r^2$ for triangle tensors.) The low rank terms of next iteration are found by running a rank revealing factorization (such as the SVD or rank-revealing QR) on $ \mU_{\ell+1} $ and $\mV_{\ell+1}$ concatenated with the low rank terms of the previous iterations and initial iterate (scaled by the appropriate $ \alpha $ and $ \beta $.) The full procedure is detailed in Algorithm~\ref{alg:lrtame}.

 The dominant terms in the overall runtime of this approach for $k$-node motifs is $O(\text{nnz}(\cmT_A) + \text{nnz}(\cmT_B)) r^{k-1} + \text{RRF}(m,(r^{k-1}+r+t)) + \text{RRF}(n,(r^{k-1}+r+t))$ where RRF is the cost of the rank-revealing factorization. There are many options for the RRF, including randomized and tall-and-skinny approaches. In our codes and the pseudocode we use a rank-revealing method inspired by the R-SVD (which does a QR factorization before an SVD to reduce the work in the SVD) and the structure of our problem. More on the asymptotic runtime of the R-SVD vs SVD can be found in \cite[Figure 8.6.1]{Golub-2013-book}.  Representative values of the ranks $r$ are typically $100$ and are much smaller than $n$ or $m$ (see more discussion in \S\ref{sec:lowrank-tame}). 


\begin{algorithm}[t]
\renewcommand{\algorithmiccomment}[1]{$\triangleright$\textit{#1}}
\caption{LowRankTAME with affine shift}
\label{LRTAME} \label{alg:lrtame}
\begin{algorithmic}[1] 
	\Require $k$-mode motif tensors $\cmT_A, \cmT_B$ for graph $A$ and $B$, mixing parameter $ \alpha $, shift $ \beta $, tolerance $\eps$, weights $\mW = \mU\mV^T$

	\Ensure Alignment heuristic $\mX$ and max-weight matching of $\mX$
	\State $C = \trace((\mV^T\mV)(\mU^T\mU)) $ \hspace{\algorithmicindent}\Comment{$C = ||\mW||^2$}
	
    \State $ \mU_0 = \mU /\sqrt{C}$;$ \mV = \mV_0 / \sqrt{C}$ \hspace{\algorithmicindent}\Comment{Normalize first iterate} 

 	\For{$\ell = 0,1, \ldots$ until $|\lambda_{\ell+1} - \lambda_\ell| < \eps$}
 	\State \Comment{Exact LowRank SS-HOPM iteration}
	\State  \Comment{Compute next iterate from low rank factors, $r = $ num cols of $\mU_{\ell}, \mV_{\ell}$}  
	\hspace{\algorithmicindent} \For{each $ i_1 $ in $1 \ldots r$, $i_2$ in $1, \ldots, r$, $\cdots$, $i_{k-1}$ in  $1, \ldots, r$ }
	\State append column $\cmT_A(\mU_{\ell}(:,i_1),\dots,\mU_{\ell}(:, i_{k-1}))$ to $ \mU_{\ell+1}$ 
 	\State append column $\cmT_B(\mV_{\ell}(:,i_1),\dots,\mV_{\ell}(:,i_{k-1}))$ to $ \mV_{\ell+1} $
	\EndFor
   \State \Comment{ Estimate tensor-eval} 	
   \State $ \lambda_{\ell+1} = \trace((\mV_{\ell+1}^T\mV_{\ell})(\mU_{\ell}^T\mU_{\ell+1})) $ 
   \State
   \Comment{ Apply affine shift in low-rank factors}
   \State $ \mU_{\ell+1} \leftarrow [\sqrt{\alpha} \mU_{\ell+1}  \quad \sqrt{\alpha \beta} \mU_\ell \quad \sqrt{1-\alpha} \mU_0] $
   \State $ \mV_{\ell+1} \leftarrow [\sqrt{\alpha} \mV_{\ell+1}  \quad \sqrt{\alpha \beta} \mV_\ell \quad \sqrt{1-\alpha} \mV_0] $
   \State \Comment{Rank-revealing factorization: Reduce to lowest rank terms}
   \State $ \mQ_{\mU},\mR_{\mU} = \tt{QR}( \mU_{\ell+1}); \mQ_{\mV},\mR_{\mV} = \tt{QR}( \mV_{\ell+1})$;
   \State $\hat{\mU}, \hat{\mSigma},\hat{\mV}^T = \tt{svd}(\mR_{\mU}\mR_{\mV}^T)$  \label{lst:RSVD-lookalike}\Comment{ Discarding near zero singular values and their vectors. }
   \State $\mU_{\ell+1} \leftarrow \mQ_{\mU}\hat{\mU}$; $\mV_{\ell+1} \leftarrow \mQ_{\mV}(\hat{\mV}\hat{\mSigma})$

      \State \Comment{Normalize} 
   \State $C = \trace((\mV_{\ell+1}^T\mV_{\ell+1})(\mU_{\ell+1}^T\mU_{\ell+1})) $
   \State $ \mU_{\ell+1} \leftarrow \mU_{\ell+1} /\sqrt{C}; \mV_{\ell+1} \leftarrow \mV_{\ell+1} /\sqrt{C}$;
   \State $\mX_{\ell + 1 } = \mU_{\ell+1}\mV_{\ell+1}^T$
   \State Set $t_{\ell+1}$ to be the number of motifs matched by a matching from $\mX_{\ell+1} $
   \EndFor
  \State \Return $\mX_{\ell}$ and the matching of $\mX_{\ell}$ with the highest $t_{\ell}$ 
\end{algorithmic}
\end{algorithm}

The primary limitation to contracting with low rank components is how much memory explicitly computing the terms requires. When $ r^k < \min\{m,n\}$, then building $\mU$ and $\mV$ is preferable because finding the low rank components for the next iteration can be done more efficiently and accurately than a dense $\mX$.  (For accuracy, see \S\ref{sec:rank-1-sing-valsAppendix}). If $r^k > \min\{\text{nnz}(\cmT_A), \text{nnz}(\cmT_B)\}$ then running the original TAME implicit multiplication procedure will be faster (as can be seen in the 7 clique results of figure~\ref{subfig:TAME_clique_scaling_detailed}). However when $\min\{m,n\} \leq r^k < \min\{\text{nnz}(\cmT_A), \text{nnz}(\cmT_B)\}$, then the matrices $\mU$ and $\mV$ become \emph{wide}. In these cases, 
the low rank structure itself is only beneficial in reducing overall work. Thus, we can simply accumulate the results treat $ \mX $ as accumulation parameter and update it with the outer product of the columns of $ \mU $ and $ \mV $ as we compute them. This can be made more efficient by computing batches of columns, but the best batch size will be system dependent and is a level of tuning we leave to end users.


\subsection{\LambdaTAME} \label{sec:lambda-tame} \label{sec:lambdatame} 

The inspiration for using SS-HOPM in TAME is that TAME's objective function \eqref{TAMEObjective} is nearby the dominant eigenvector problem for $ \cmT_B \otimes \cmT_A $. Given the observation in theorem~\ref{thm:spectrum-general} that the dominant eigenvector is built from the dominant eigenvectors of $ \cmT_B $ and $ \cmT_A $, this suggests a new heuristic which can be run using only the tensor powers sequences of $ \cmT_B $ and $ \cmT_A $ independently, rather than combining them as is done in TAME. We then store each of the iterates into a pair of matrices $ \mU $ and $ \mV $ and use the information in $\mU$ and $\mV$ to derive the matching. There exist many possible ways to derive a matching from the iterates stored in $\mU$ and $\mV$ (see~\cite{nassar2018low} for many low-rank ideas). We found that performing a max-weight matching on $\mX = \mU \mV^T$ was the most accurate for downstream alignment tasks in our initial investigation. This is a heuristic choice. Our only ad-hoc justification is that, if these had been matrices, this would have been a set of inner-products among the Krylov basis.  We discuss additional useful refinement of $\mU$ and $\mV$ in the next section.
 We call this method \LambdaTAME because it is inspired by our dominant Z-eigenvalue theorem.
\begin{algorithm}[t!]
\caption{\LambdaTAME} 
\label{LambdaTAME} \label{alg:LambdaTame}
\begin{algorithmic}[1] 
    \Require $k$-mode motif tensors $\cmT_A, \cmT_B$ for graph $A$ and $B$; mixing parameter $ \alpha $, shift $ \beta $, max iterations $ L $
    
	\Ensure Alignment heuristic $\mX$ and max-weight matching of $\mX$
    \State $ \mU(:,1) = \frac{\mathbbm{1}_m}{\sqrt{m}}; \mV(:,1) = \frac{\mathbbm{1}_n}{\sqrt{n}}$ \hspace{1cm}\Comment{Initialize first columns}
    
    \For{$\ell=1, \ldots, L$} 
	    \State $\mU(:,\ell\!+\!1) = \cmT_A \mU(:,\ell)^{k-1}; \mV(:,\ell\!+\!1) = \cmT_B \mV(:,\ell)^{k-1} $
	    \State $ \mU(:,\ell+1) \leftarrow \alpha \mU(:,\ell+1) + \alpha \beta  \mU(:,\ell) + (1-\alpha) \mU(:,1) $
	    \State  $\mV(:,\ell+1) \leftarrow \alpha \mV(:,\ell+1) + \alpha \beta \mV(:,\ell) + (1-\alpha) \mV(:,1) $
	    \State $ \mU(:,\ell+1)  = \frac{\mU(:,\ell+1) }{\|\mU(:,\ell+1)\|};\quad \mV(:,\ell+1) = \frac{\mV(:,\ell+1) }{\|\mV(:,\ell+1)\|} $
    \EndFor
	\State \textbf{Return} $\mX = \mU \mV^T$ and the matching  from $\mX$ 
%
%
\end{algorithmic}
\end{algorithm}

Again, we adopted an affine-shift variant of the TAME method that includes an $\alpha$ factor to \emph{reintroduce} the original vector into the solution. This can be set to $1$ so that the iterates are exactly those from the SS-HOPM method, but there are cases where $\alpha \not=1$ helps. In the algorithm, both $ \mU $ and $ \mV  $ can be computed in time proportional to the number of non-zeros of their tensors times the total number of iterations.  Like LowRankTAME, the computational bottleneck of this algorithm becomes the matching and refinement steps (see figure~\ref{subfig:LVGNA_TTVMatchingRatios}). 

\subsection{Matching Refinement}
\label{sec:refinement}

 Refining the final matching is a necessary addition when using either TAME, LowRankTAME, and \LambdaTAME. For each method, the result is both a low-rank matrix $\mX^*$, along with the rank factors $\mU, \mV$, and a maximum weight matching computing on this matrix. 
 In the original TAME method, the matrix $\mX$ was improved by computing a maximum weight bipartite matching and then by looking locally for potential match swaps which montonically increase triangles aligned.   Another approach using with a low-rank method is to use the information and matching produced to initialize and guide a more expensive network alignment method, such as Klau's algorithm~\cite{klau2009new}, similar to what was done in \citet{nassar2018low}. 
 
 \textbf{TAME's $b$-matching local search refinement} 
 	TAME's refines its produced matching by constructing local neighborhoods of nodes and looking for substitutes in its current matchings that increase the number of triangles aligned (or increases the number of edges while maintaining the triangles aligned.) The authors construct a $ b $-matching from the matrix $ \mX^* $ returned by TAME (using the 2-approximation algorithm~\cite{khan2016efficient}) and search the found matchings along with neighbor substitutions as local neighborhoods. Each edge $ (i,i') $ in the matching, in order of their edge weight, searches the set of alternative matches
 	\[
 	 \bigg\{(i,j') \bigg|  \begin{array}{c}
 	 	(i,j') \in \texttt{b-matching(}\mX^*\texttt{)}, \text{ or } \\
			j' \text{ is connected to } i' \text{ in graph $B$ } \\
 	 \end{array}\bigg\}\cup
  	 \bigg\{(j,i') \bigg|  \begin{array}{c}
  		(j,i') \in \texttt{b-matching(}\mX^*\texttt{)}, \text{ or }  \\
		 j \text{ is connected to } i \text{ in graph $A$ }\\
  	\end{array}\bigg\}
 	 \]
	for a possible replacement, and immediately makes changes which improve the alignment. The full procedure is outlined in~\cite[section 4.5, Algo. 4]{mohammadi2017triangular}. The original method ascribes weights the edges and triangles using the weights in the iterate returned by TAME, but our method doesn't weight the triangles or edges when measuring the change in alignment quality. The greedy swapping procedure can be run multiple times, but improvements tend to stop after 5-10 successive sweeps over all matched edges.
	
	\textbf{A new nearest neighbor local search refinement.} The low rank structure of $ \mX^* $, suggests that an alternative to $b$-matching, for which even the 2-approximation is computationally costly on a large, dense matrix $\mX^*$.  Rather than $b$-matching, we treat the low-rank structure $\mX^* = \mU \mV^T$ as an \emph{embedding} of each vertex  where rows of $\mU$ give coordinates for each vertex in graph $A$ and rows of $\mV$ give coordinates for each vertex in graph $B$. Then we consider nearby vertices as alternative matches. For this task, a $ K $ nearest neighbors methodology applies. Each row of $ \mU $ embeds $ i \in V_A $, so the rows of $ \mU $ which are close to $ \mU(i,:) $ in $2$-norm distance define a natural neighborhood of $ i $. This leads us to construct sets of the form 
     	\[
    \bigg\{(i,j') \bigg|  \begin{array}{c}
    	j' \in K\texttt{-nearest(}\mV(i',:),\mV\texttt{)}, \text{ or } \\
    	j' \text{ is connected to } i' \text{ in graph $B$ }\\
    \end{array}\bigg\}\cup
    \bigg\{(j,i') \bigg|  \begin{array}{c}
    	j \in K\texttt{-nearest(}\mU(i,:),\mU\texttt{)}, \text{ or } \\
    	j \text{ is connected to } i \text{ in graph $A$ }\\
    \end{array}\bigg\}
    \]
    to search for changes to the matchings.  Ball-trees are particularly suitable for finding close neighbors of points in low dimensional spaces and are empirically faster than $b$-matching with superior results.
    
 \textbf{Improving matchings with Klau's algorithm.} Klau's algorithm \cite{klau2009new} is an edge based graph matching / network alignment method that uses a sequence of maximum weighted matchings to iterate towards a better solution. It can, in some instances, identify optimal solutions of the NP-hard graph matching objective with a corresponding proof of optimality.  The algorithm is built from a Lagrangian decomposition of a tight linear program relaxation of the graph matching IQP (a weighted form of Def.~\ref{GlobalGraphAlign}). A full explanation of the algorithm can be found in \cite[section 4.3]{Bayati-2013-netalign}. The primary input for Klau's method are the graphs $A$, $B$ and a weighted bipartite graph between the vertex set of $A$ and $B$ that restricts and biases the set of possible alignments. The adjacency matrix of this bipartite graph is $ \mL $ and is called the \textit{link matrix} or prior matrix. The method is most effective when $\mL$ has only a few choices for alignments between the graphs. 
 
  Thus, we use the results of LowRankTAME or \LambdaTAME to build $\mL$. We include the matched edges within $ \mL $ and then expand using the neighborhoods of the matched nodes (much like TAME's local search). In \citet{nassar2018low}, Klau's method was more accurate when given expanded results of b-matching. 
   Given the low rank structure of our methods, we further expand $\mL$ by including edges in the found matchings with the $ k $ closest neighbors of $ (i,i') $ in their respective embedding spaces $ \mU $ and $ \mV $.

\section{Empirical Comparisons in our Network Alignment application}
\label{sec:experiments}
The major demonstration of the new Kronecker product theory is in terms of its impact on network alignment algorithms described in the previous sections. We have implementations of TAME which compute contractions using the original implicit form and new versions using our and existing tensor Kronecker theory.  These are all generalized to work with any order motif. We focus on cliques as the motif. We use TuranShadow~\cite{jain2017fast} to sample the network for cliques at random. Equivalently, we use cliques to induce a hypergraph where the nodes are the same and the cliques are hyperedges. Our codes are implemented in Julia and are available from {\footnotesize \texttt{\url{https://www.cs.purdue.edu/homes/ccolley/project_pages/TensorKroneckerProducts.html}}}.  In this section, we validate the algorithms and show we can achieve similar results with greatly improved runtimes. Some highlights of our results:
\begin{enumerate}
	\item Iterates of TAME are low rank on real and synthetic data and LowRankTAME computes them an order of magnitude faster for small enough motifs (\S\ref{sec:lowrank-tame}). 
	\item The \LambdaTAME vector information can be produced quickly for any size motif. 
	\item When the \LambdaTAME vector information is refined using the nearest neighbor information and Klau's algorithm, it aligns more triangles and edges than the refined TAME information. Also, it has end to end runtimes 1-2 orders of magnitude faster than the C++ TAME implementation (\S\ref{sec:biological-networks}). 
\end{enumerate}

 We use all the same parameters as the original research where they were accessible and will discuss our reasoning for our choices for unlisted parameters. Our experiment environment uses Intel Xeon Platinum 8168 CPUs (@ 2.70GHz) processors with 24 cores, although none of our methods use multicore parallelism. 
	We compare our methods against one another as well as LowRankEigenAlign~\cite{nassar2018low}. LowRankEigenAlign utilizes low rank structure discovered in the EigenAlign~\cite{feizi2019spectral} algorithm and improves its scalability with minimal changes or even sometimes improvements to accuracy.  LowRankEigenAlign has been tested on similar real world and synthetic alignment problems, and its low rank structure makes it a comparable method in terms of memory to \LambdaTAME. LowRankEigenAlign also gives a low rank embedding which allows us to refine its results in the manner similar to \LambdaTAME and LowRankTAME.
\subsection{Data for network alignment experiments}
\label{sec:netalign-data}
There are two types of data that we use in evaluating the new network alignment algorithms. The first is a subset of the LVGNA~\cite{vijayan2017multiple} protein-protein interaction (PPI) graph collection. 
Each pair of networks in this collection gives an alignment problem. Network statistics are in the supplemental materials (Table~\ref{tab:LVGNA}). Each vertex represents a protein and the edges represent interactions. The networks range in size from 2871 to 16060 vertices and all but the largest networks have fewer triangles than edges. 

The second type of data involve synthetic random geometric (RG) graphs. To generate a RG graph, we randomly sample $n$ points in the unit square. Then each point adds undirected edges to the $k$ nearest neighbors, where $ k $ is drawn from a log-normal distribution centered at $ \log{5} $ with  $\sigma= 1 $. We then create a pair of networks to align from this starting reference graph by independently perturbing them from a noise model. The two noise models we consider are (i) a microbiological inspired partial duplication procedure~\cite{bhan2002duplication,chung2003duplication,hermann2014large} and (ii) the \ErdosRenyi (ER) noise model from~\cite[section 3.4]{feizi2019spectral}. When using the duplication noise mode, we incrementally duplicate 25\% new nodes in the network, copying their existing edges with probability $ p_{edge} = .5$. For the ER noise we randomly delete edges with probability $ p = .05 $ and randomly add in edges with probability $ q = \frac{p\rho}{1 - \rho} $, where $ \rho $ is the density of the network. A few experiments have different choices for these parameters, which will be  explicitly noted. We further randomize the permutation of the perturbed network to avoid any influences due to node order in what might happen in the presence of tied values. (Prior work and experience has shown a startlingly strong effect due to biases when this permutation step is not present.) Within each experiment, algorithms are always tested on exactly the same set of networks instead of separate draws from the same distribution. 

\subsection{Low-rank structure in TAME}
\label{sec:lowrank-tame}

For our first set of experiments, we want to show that the iterations from TAME (Algorithm~\ref{alg:tame}) remain low rank when we start with a uniform, or unbiased iterate as the weight matrix: $\mX_0  = \frac{1}{mn}\ones_{m}\ones_{n}^T$, which is rank 1. We further investigate this behavior on larger motifs. To do this, we report matrix rank using the LowRankTAME algorithm instead of the raw TAME algorithm, which are identical in exact arithmetic (see Aside~\ref{aside:lrtame-vs-tame}).
\aside{aside:lrtame-vs-tame}{This choice of exact LowRankTAME vs.~TAME to evaluate rank is made both because it is faster to compute but also because preliminary experiments showed that TAME caused the finite precision rank to grow even when the result is mathematically rank 1 ($ \alpha = 1.0 $, $ \beta =0.0 $, by lemma \ref{lmm:rank1}). This is well-known to happen to finite precision computations, for instance in the power method. Details of this test are in the Supp.~\ref{sec:rank-1-sing-valsAppendix}.} Given either a pair of networks from LVGNA or a synthetic alignment problem, we plot the maximum rank from any iteration on any trial ($ \alpha \in \{.5,1.0\} $ running 15 iterations) as determined by the rank function in Julia, as we vary the $\beta$ parameter in the alignment problem. (See Figure~\ref{subfig:maxRankExps}.) These results, along with trendlines for the maximum rank over multiple repetitions of the synthetic experiments, show that the rank is often below $250$ even though the largest networks have 10k vertices.

\begin{figure}[ht]
	\begin{minipage}[m]{.49\textwidth}
		\centering
		\includegraphics[width=1\textwidth]{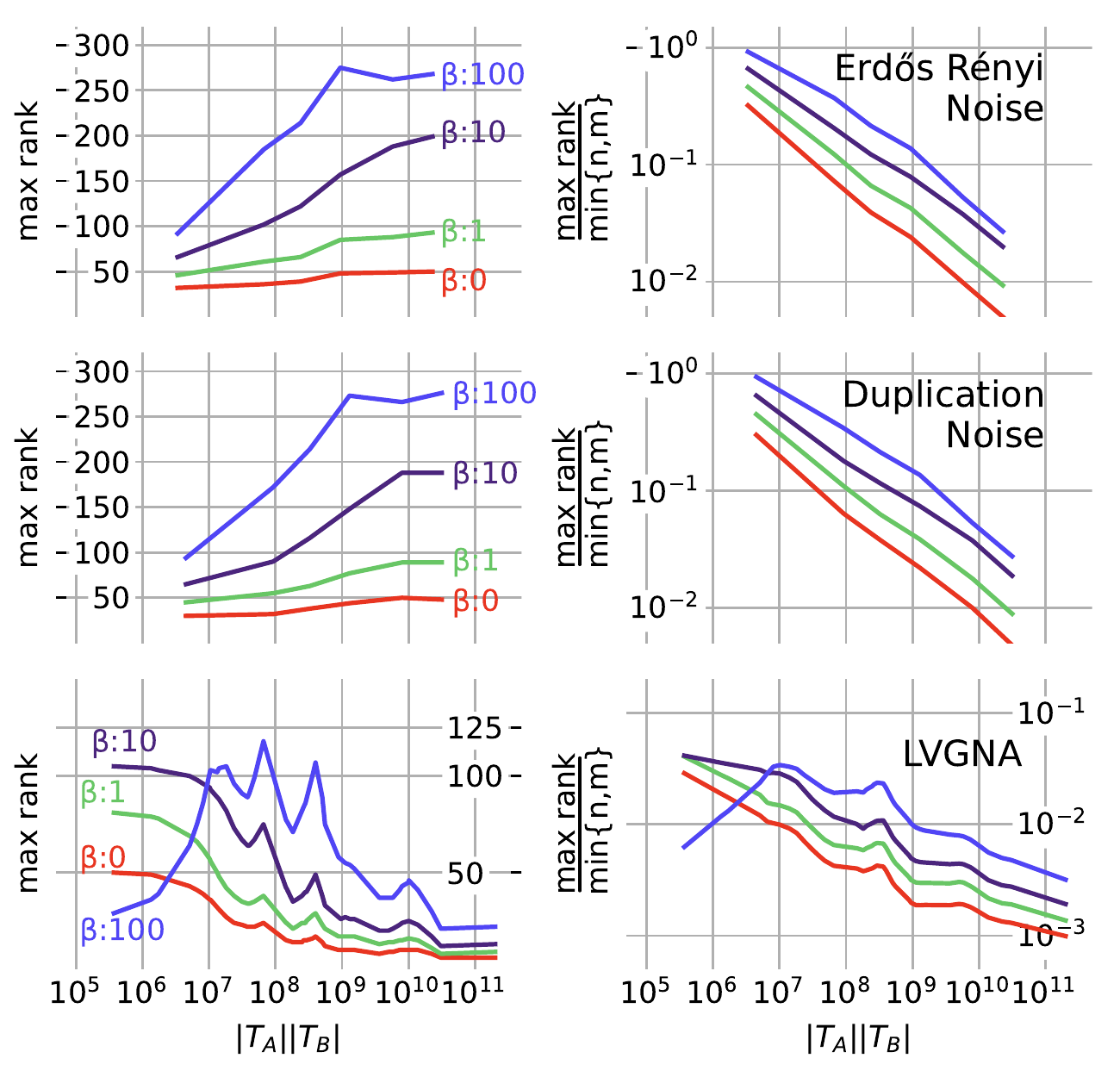}
		\vspace{-7.5mm}
		\noindent \caption{\label{subfig:maxRankExps}
			Both real-world and synthetic results (reference graphs are generated with 100, 500, 1000, 2000, 5000, 10000 vertices) are low-rank with respect to the size of the networks. We compute the maximum rank over any iterate from runs with any the affine shift values $ \alpha=0.5, 1.0 $, and we plot the maximum rank directly for the synthetic networks and loess smoothing trendlines (using 30\% approximate neighbors) for the LVGNA experiments. The maximum rank of any iterate over synthetic network alignment problems were consistently higher than PPI problems, but both are low when put in the context of their maximum possible ranks (right hand plots). The similarity of results between the two noise models is expected as they start with the same reference graph. We generally see that rank increases as $\beta$ increases except for $\beta=100$, which is discussed in the text.}
	
	\end{minipage}
	\hfill
	\begin{minipage}[m]{.49\textwidth}
		\centering
		\includegraphics[width=1\textwidth]{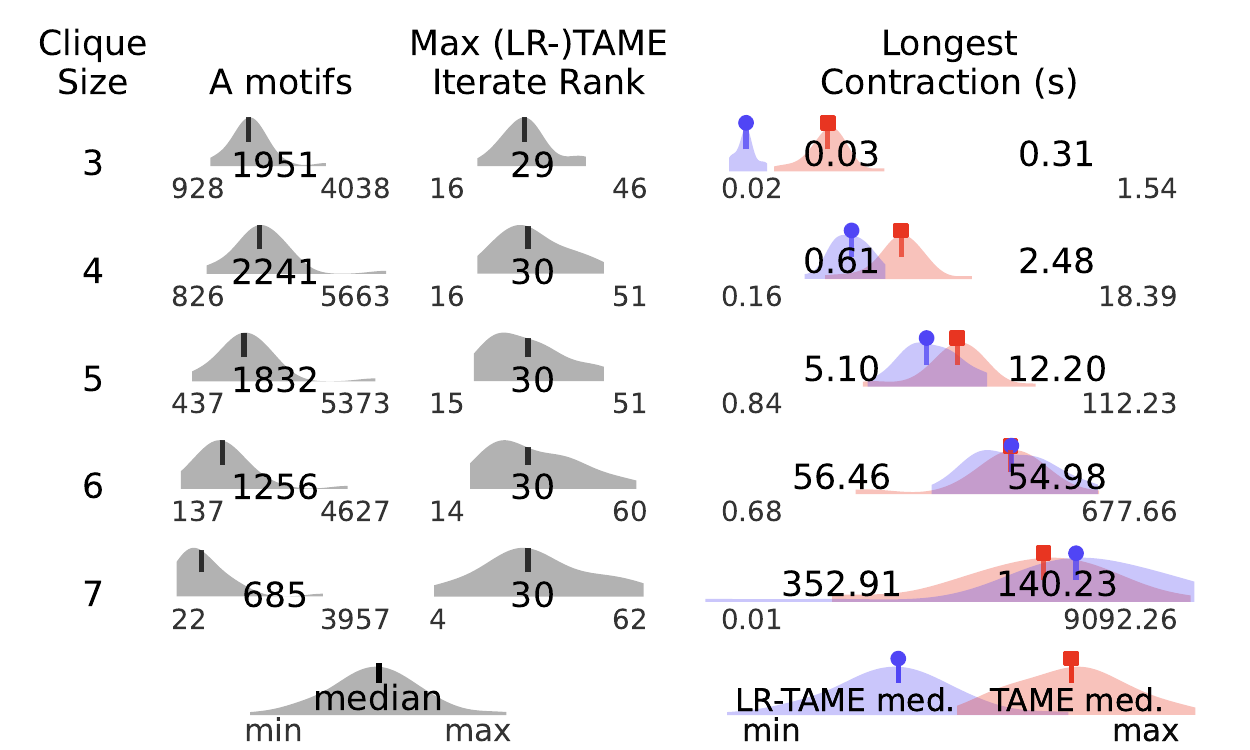}
		\vspace{-7.5mm} 
		\noindent \caption{\label{subfig:TAME_clique_scaling_summarized}
			Experiments on rank from synthetic experiments where the reference graph has $100$ vertices and is perturbed with 20\% duplicated nodes (instead of the default 25\%). We compute statistics over 25 trials and 15 iterations of each method (LR-TAME for LowRankTAME and ``med.'' for median).  The figures show density plots of the worst results over all the trials. A more detailed analysis of our data can be found in Figure~\ref{fig:FExpsrank-results}.  Exploiting the low rank structure is most effective for small motifs. We see that time spent computing contractions for TAME and LowRankTAME grows as the motif size increases, even though the rank of TAME's iterates and number of motifs declines. As a point of comparison, \LambdaTAME's runtime is reasonably constant across each experiment and the longest \LambdaTAME contraction time of any trial was 0.0132s.}
	\end{minipage}
\end{figure}

Our rank experiments show the synthetic problems have higher ranks than the LVGNA collection. For the LVGNA collection, large problems tend to have smaller rank, whereas we do see the rank grow with the size of the synthetic problems. For the synthetic problems, we also see that increasing $\beta$ produces higher ranks because these problems incorporate more of the previous iterate via an affine shift. The behavior with $ \beta =100 $ for the LVGNA collection is rather different, with many small dips. On further investigation, we found this occurs because $\beta = 100$ is nearly an eigenvalue of these problems. We verified in subsequent experiments that shifting by the estimated eigenvalue gives very small rank, although we do not report these experiments in the interest of space. 


In Figure~\ref{subfig:TAME_clique_scaling_summarized}, we investigate rank behavior for larger clique motifs.  We see that though the rank of iterates does not change dramatically (second column of density plots), the runtime of TAME and LowRankTAME consistently grows (blue and red density plots in third column), even when the number of motifs within the networks decline (first column of density plots). 
As the motif size grows, the time spent using the low rank contraction routines approaches the runtime of TAME's implicit contraction. This becomes salient when memory constraints require a user to use the accumulation form of LowRankTAME, as then even the low rank components are found from a dense matrix $ \mX_{\ell} $, rather than being able to benefit from two R-SVD calls. In summary figure~\ref{subfig:TAME_clique_scaling_summarized} indicates that for small motifs (less than size 6), we can improve the runtime and accuracy using the contraction theory shown in this paper, but those benefits are reduced or even eliminated for problems with larger motifs.

\subsection{Alignment Accuracy in Synthetic Networks}
\label{sec:synthetic-networks}	
The next set of experiments transitions from runtime to accuracy where we test how well the best low rank results produced by the TAME method, \LambdaTAME, and LowRankEigenAlign can be refined by local search and Klau's algorithm using the $ K $-nearest neighbor strategy. We focus on the synthetic problems where there is a single reference graph that is subject to two independent perturbations. The goal is to find the alignment between the vertices of the original reference graph, which we regard as the \emph{correct} answer. Each combination of methods is compared using the accuracy 
\[ \text{accuracy} = \frac{\text{number of aligned pairs of vertices from the reference graph}}{\text{total number of vertices in the reference graph}} \]
and their triangle alignment score (how many triangles they match compared to the maximum possible). 
 We use max iterations $L = 15$, stopping tolerance $\varepsilon = 10^{-6}$, and  $ K = 2*\texttt{rank(}\mX^*\texttt{)} $ throughout the experiments -- except in figures where $K$ is varied. We focus on our experiments which vary the size of alignment problems. Additional parameters of our noise models are studied in the supplement~\ref{sup:FExpsRandomGraphModels}. 

The first set of experiments focuses on triangles. These experiments show that all three methods require refinement to get practical results, especially as the problems get larger. (figure~\ref{subfig:nSizeSynthExp}). These experiments show that LowRankTAME with the local search strategy $K$-NN had the best performance for the largest problems, although \LambdaTAME with Klau's refinement was slightly better at intermediate sized problems for the duplication noise model. 
 
We can also see that triangles matched is a good proxy for the accuracy of the matchings, although depending on the noise models, there may be deviations.

	\begin{figure*}[t]
		\begin{minipage}[m]{.39\textwidth}
			\centering
			\includegraphics[width=1\textwidth]{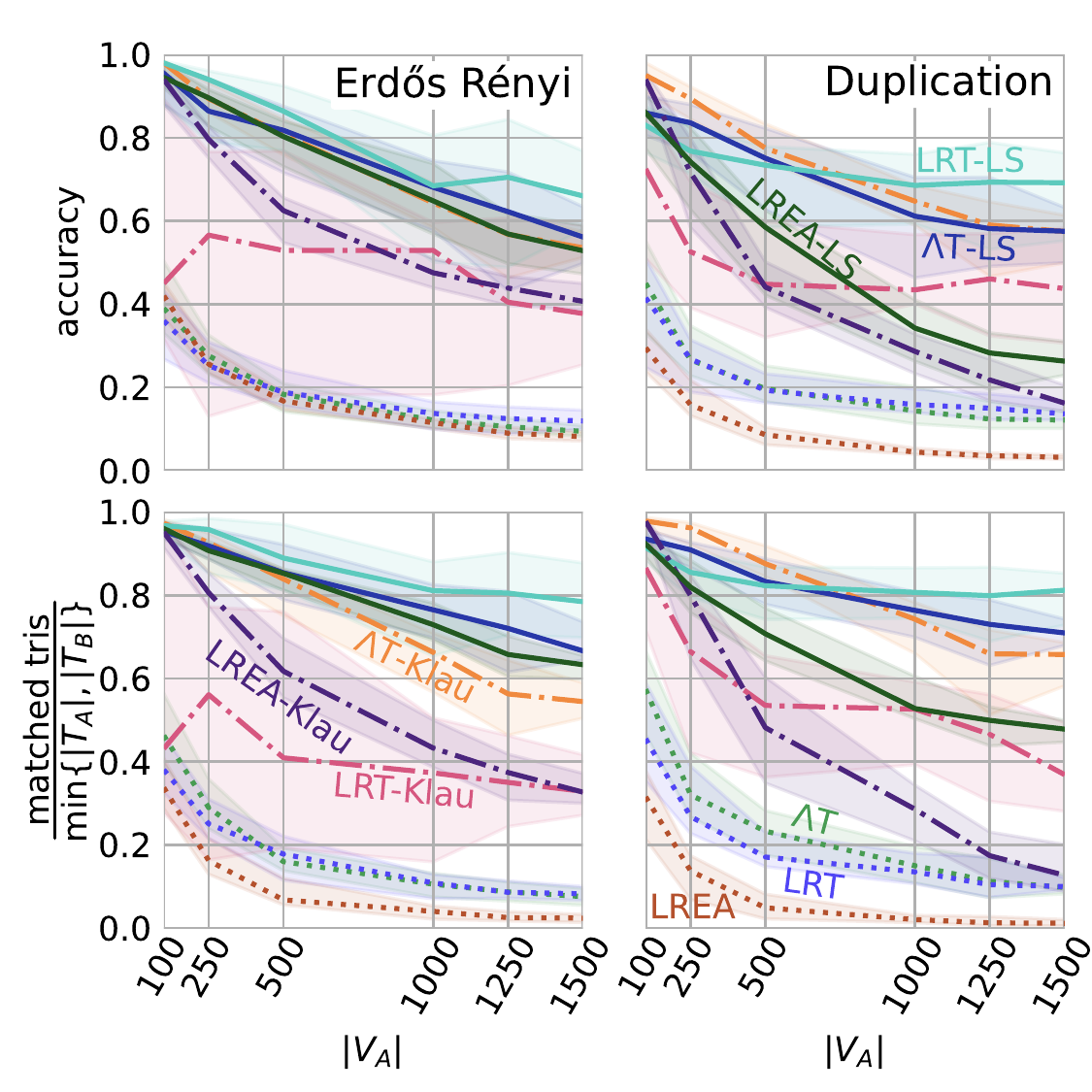}
			\vspace{-7.5mm}
			\caption{\label{subfig:nSizeSynthExp}
				We consider aligning two independent perturbations of a single reference graph using either the \ErdosRenyi (left column) or duplication noise model (right column) based on matching triangles. We compare three methods: LowRankTAME (LRT), \LambdaTAME ($\Lambda$T), and LowRankEigenAlign (LREA) with three refinement schemes: None, Klau, and local search (LS). Across all methods, the ground truth accuracy (top row), which is generally not known, is closely aligned with the number of matched triangles (bottom row), which is easy to compute, suggesting that the latter is a useful proxy. 
				This shows that refinement is an important step as the methods without refinement have dramatically lower accuracy (fine dots) than either local search (solid lines) or Klau (dash-dots).
				We plot the median of 20 trials with 20th-80th percentile ribbons.  }
			
		\end{minipage}
		\hfill
		\begin{minipage}[m]{.59\textwidth}
			\centering
			\includegraphics[width=1\textwidth]{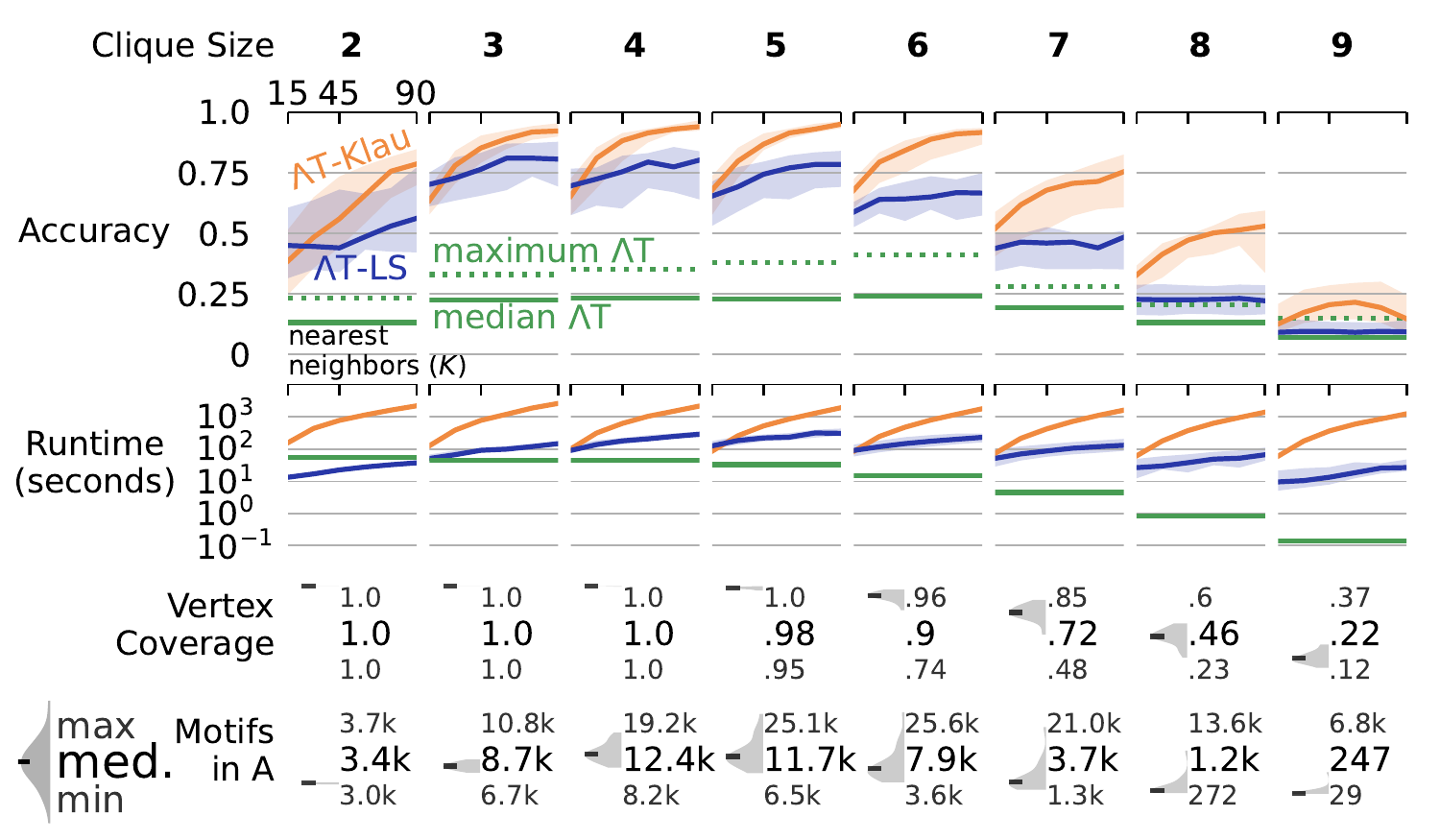}
			\vspace{-7.5mm}
			\caption{\label{subfig:K_nearest_experiments}
				We consider the same scenario for figure~\ref{subfig:nSizeSynthExp} but now look at aligning networks based on cliques of size 2 (edges), 3 (triangles), up to size 9 on networks with $500$ vertices in the duplication model. Cliques are sampled using TuranShadow with $ 10^6 $ samples, which will find the vast majority. We focus on the \LambdaTAME ($\Lambda$T) method as LowRankTAME would take a prohibitively long time (days). We also vary the number of nearest neighbors considered in the refinement step for both Klau and local search (LS) in the horizontal piece of the microplots to understand that behavior, as well as its impact on runtime (2nd row). The top row (accuracy) shows that accuracy declines after the clique size is larger than 5 or 6 for either refinement strategy. To understand this behavior, we look at the total number of motifs found (bottom row) and the vertex coverage of those motifs (3rd row). These are shown as density plots with the max, min, and median values shown.  This shows that accuracy declines once the vertex coverage begins to decline. 		}
		\end{minipage}
	\end{figure*}

	Moreover, with the $K$-NN refinement and \LambdaTAME's scalability, we can get fast accurate matchings for not only large networks, but also increasingly larger clique sizes (figure~\ref{subfig:K_nearest_experiments}). Accuracy remains high when the vertex coverage, the fraction of the total vertex set involved in motifs, remains high. In contrast to the results with LowRankTAME (figure~\ref{subfig:TAME_clique_scaling_summarized}), using \LambdaTAME has a practical runtime for large cliques. 
	Klau's algorithm can offer an additional benefit if a longer runtime can be tolerated. 
	
	We also find that our default choice of $ K $ gets good performance. Increasing $ K $ can improve accuracy slightly, but Klau's algorithm runtime is more sensitive to the sparsity of the input link matrix. Local search remains very fast with a modest increase in runtime as the local search neighborhoods are expanded. It's unsurprising to see that Klau's algorithm matches the fewest motifs given the objective function is focused on aligning edges between networks.

	\subsection{Biological Networks}
	\label{sec:biological-networks}
	We now turn our attention to considering the performance within real world networks from biology.    We report the end to end runtime, triangles matched, and edges matched of each refined method relative to TAME, over pairs of alignment problems from LVGNA in figure~\ref{subfig:LVGNA_postProcessing} (see \S\ref{sup:LVGNA_breakdown} for non-comparative results). We use max iterations $L = 15$, stopping tolerance $\varepsilon = 10^{-6}$, and nearest neighbors $K = 2*\texttt{rank(}\mX^*\texttt{)}$ for refinement.  We see that \LambdaTAME refined with local search aligns more triangles and edges and runs much faster than the original TAME. All methods tested align more edges than TAME, though improving with local search was much more likely to increase triangles matched. Methods using Klau's algorithm or LowRankEigenAlign increased the number of edges aligned, but aligned fewer triangles. This was observed in the synthetic results in~\ref{subfig:nSizeSynthExp} and is unsurprising given each method's focus on edges.

	\begin{figure*}[t]
		\begin{minipage}[m]{.59\textwidth}
			\centering
			\includegraphics[width=1\textwidth]{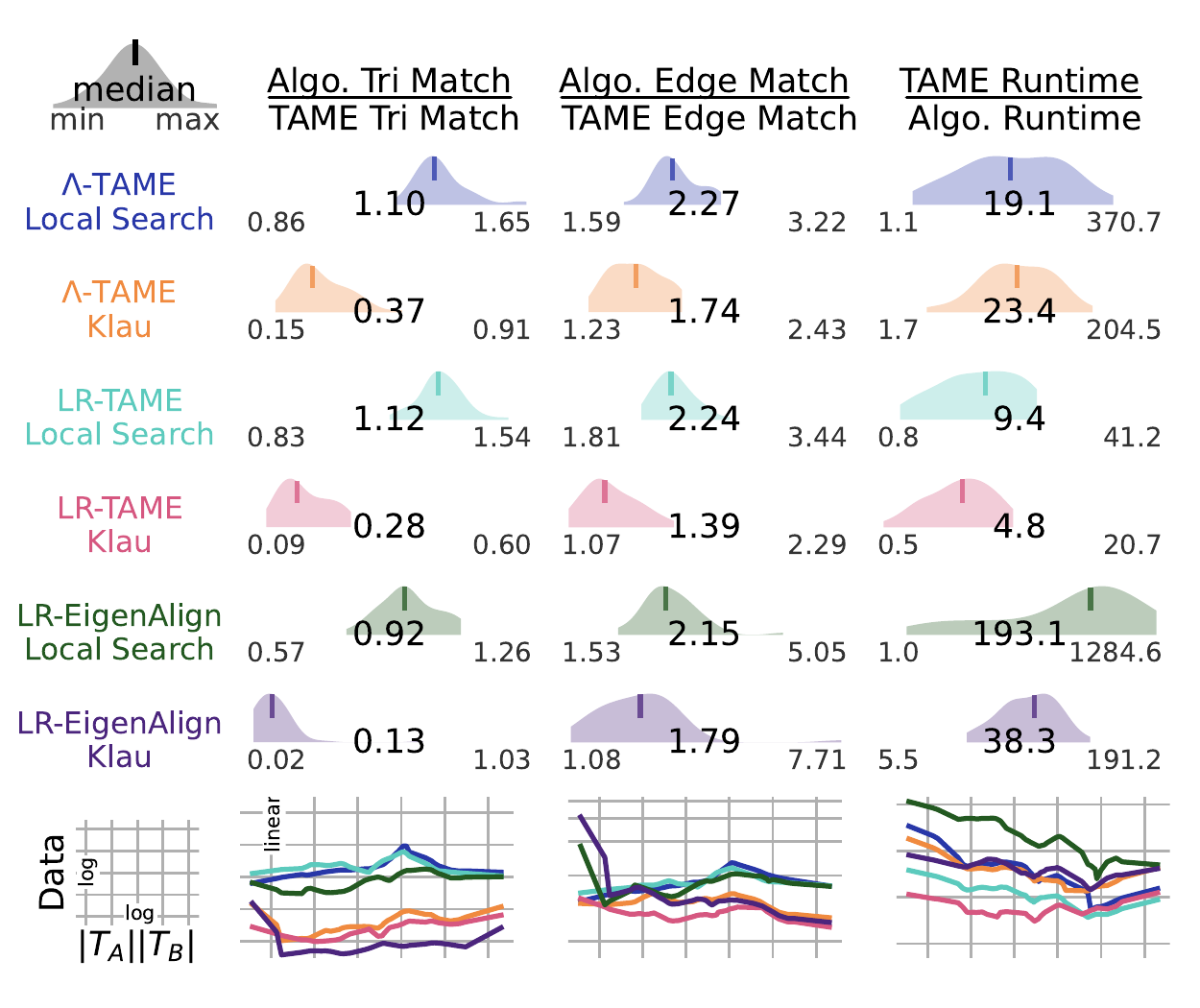}
			\vspace{-7.5mm}
			\caption{\label{subfig:LVGNA_postProcessing}
				For networks in the LVGNA collection, we compare the number of triangles (left column), edges (middle column), and runtime (right column) between the low-rank methods and the original TAME method (including its end-to-end $b$-matching refinement time). These are shown as density plots over all 45 pairwise alignment problems. Larger values and values larger than $1$ are \emph{better} for all experiments. The final row shows the Loess-smoothed plot of the raw data against the problem size, which shows minimal size-dependent effects -- beyond those expected due to runtime. Note that  \LambdaTAME with local search consistently aligns more triangles and edges than TAME while running about 20-times faster. Refining with LocalSearch tends to be faster than using Klau's algorithm though we expect the sparsity of the input matrices to be the same. These experiments show that refinement can be very problem dependent and local search is particularly successful here.}
		\end{minipage}
		\hfill
		\begin{minipage}[m]{.39\textwidth}
			\centering
			\includegraphics[width=1\textwidth]{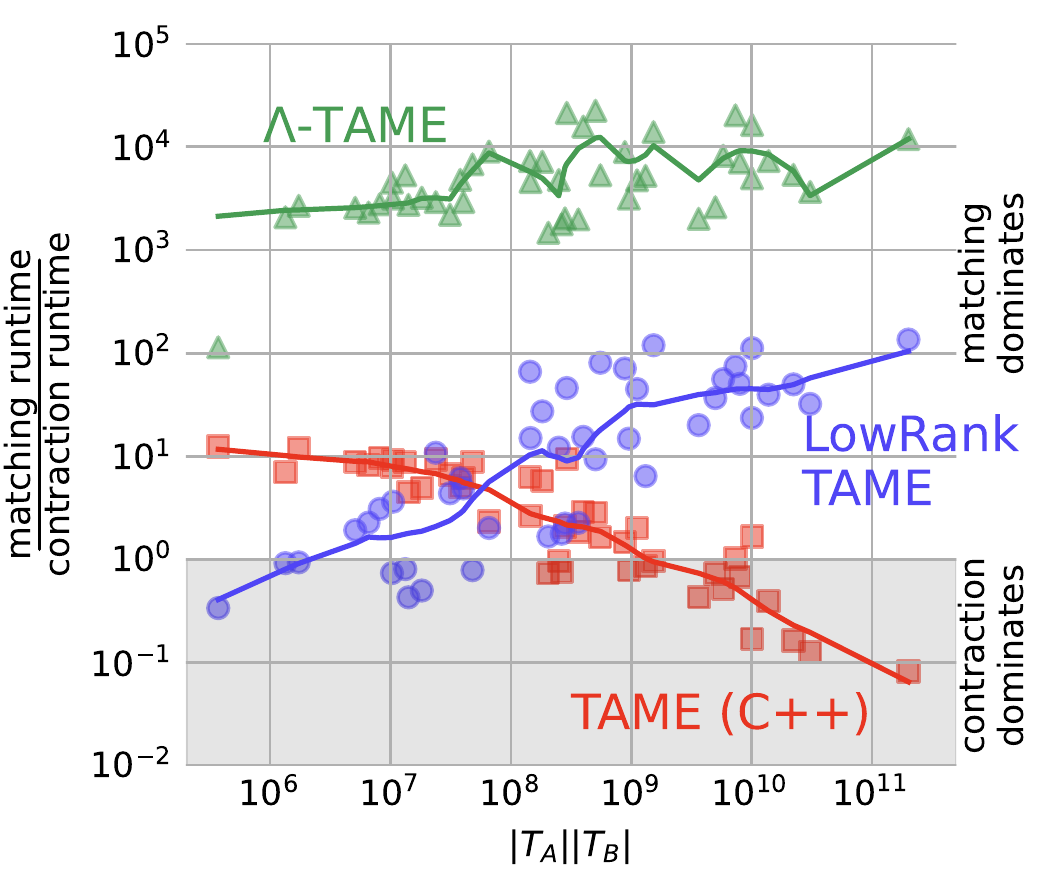}
			\vspace{-7.5mm}
			\caption{\label{subfig:LVGNA_TTVMatchingRatios}
				We compare the time spent working on tensor-vector multiplication / contraction compared with the time spent rounding $ \mX_{\ell} $. These show that the new bottleneck of \LambdaTAME and LowRankTAME is the time spent on rounding the continuous iterates to discrete matchings, in contrast to the original TAME method.}
		\end{minipage}
	\end{figure*}

 Returning to timing, we compare the fraction of time spent in the matching vs.~matrix / tensor operations in figure~\ref{subfig:LVGNA_TTVMatchingRatios}.  We can see that the Kronecker theory here makes the time to compute the contractions fast enough to change the primary bottleneck of the algorithm when using triangle adjacency tensors. Put simply,  \LambdaTAME always spends more time on the bipartite matching and refinement and LowRankTAME spends more time there on the biggest problems. A primary reason for why \LambdaTAME and LowRankTAME spend so much time computing the matchings is that we cannot take advantage of the low rank structure. We compute the maximum matching using the Primal-Dual algorithm~\cite{dantzig1956primal}, which must touch each entry of $ \mX^* =\mU\mV^T $ at least once, making it more efficient to form $ \mX^* $ explicitly at the beginning of the algorithm. This suggests potential future research for new algorithms which can properly make use of the low rank structure, while still computing a maximum weighted matching.

\section{Discussion}
\label{sec:related}


The major focus of our paper is on demonstrating how the theory on tensor Kronecker products in \S\ref{sec:Kronecker} enables us to accelerate the graph matching algorithm TAME (\S\ref{HOAlgos}): (i) by making the same algorithm faster with LowRankTAME, (ii) by giving a new, faster algorithm (\LambdaTAME), and (iii) by providing new augmentations of TAME's local search and Klau's algorithm which can make use of the low rank structure within the iterates. 


One interesting theory question we have not pursued is the opposite of the example from the introduction that shows diagonal tensors have eigenvectors that are not a Kronecker product. Put concretely, is there a class of tensors where no new eigenvectors emerge after taking a Kronecker product? 

On the application side, the new method \LambdaTAME's runtime is heavily dominated by the rounding and refinement procedures, as seen in figure~\ref{subfig:LVGNA_TTVMatchingRatios}. Our implementation uses the primal-dual algorithm which is an effective solution when $ \mX  = \mU \mV^T$ is explicitly realized as a dense matrix. New matching methods which compute the maximum matching of $ \mX $ while only using $ \mU $ and $ \mV $ would be useful to improve scalability (even if only an approximation). For large enough problems, there are also low-rank matching heuristics from~\cite{nassar2018low} to consider for additional scalability, although the results from these methods were noticeably worse for our case compared with using the exact max-weight matching.


Low-rank structure offers a few benefits even beyond the reduced runtime. First, we are able to explicitly \emph{store} a large number of TAME iterates as low-rank factorizations. Sending $ O(mr) $ data is much faster to send between cores, which may offer footholds in known parallel matching challenges~\cite{Sathe-2012-matching,Bertsekas-1991-auction}. Applying theorem~\ref{thm:spectrum} recursively suggests an immediate algorithm for multi-network alignments. Each network's embeddings can be computed independently in a fashion similar to that used in~\cite{Nassar-2021-multinetwork}. Furthermore, we also see opportunities to incorporate multiple network motifs within adjacency tensors. Smaller motifs could be encoded into in the off diagonal components (see Aside~\ref{aside:multi-motifs}).  Motif complexes could be encoded into the off diagonal components in a way that wouldn't change contraction or eigenvector definitions. 

\aside{aside:multi-motifs}{Tensors can have more than one \emph{``diagonal''} by grouping non-zeros by the multiplicity of their indices. In a triangle adjacency tensor the non-zeros are of the form $ (i,j,k) $ for distinct vertices and the traditional diagonal is comprised of the indices $ (i,i,i) $. A third order tensor also has entries which only have two unique vertices, and the presence of an edge could be marked in an entry of the form $ (i,i,j) $ or $(i,j,j) $. These off diagonals are referred to as $q$-multiplicity tensors in \cite[Def. 2]{yan2015discrete}. }
The fashion in which we construct the embeddings is also closely related to various graph kernels~\cite{vishwanathan2010graph,Kriege-2020-graph-kernels} including the random walk kernel on a direct product graph.  Graph kernels have long been used to align small chemicographs (graphs that represent small chemical molecules).  In this case, we are able to generate a direct factorization of a graph kernel between vertices of two graphs into a product of features on each graph. This is a common paradigm~\cite{vishwanathan2010graph} involving matrix Kronecker products---although we are unaware of any research on this for higher order analogues of the graph kernels involving tensor Kronecker products that would be needed for our perspective. When viewed in this light, our research has the potential to open new directions in this space in terms of efficient graph kernels on hypergraphs.

In summary, our theory and experiments show how the computational demands of methods with tensor Kronecker products may be reduced by orders of magnitude with no change in quality, or accelerated even further with useful approximate results.  We are excited about the future opportunities with tensor Kronecker products due to the widespread use of matrix Kronecker products, and suspect that these theorems, or alternative generalizations that use specific structure in novel problems, will be a key element in this future research.

\section*{Acknowledgments}
	We thank O. Eldaghar for the fruitful conversations when designing our figures and C. Cui for discussions on the code from~\cite{cui2014all}.

\appendix
\section{Additional information}
	\subsection{TAME rank-1 Singular Value Experiments}
	\label{sec:rank-1-sing-valsAppendix}
	\begin{figure}[h]
	\centering
	\includegraphics[width=4.25in,valign=b]{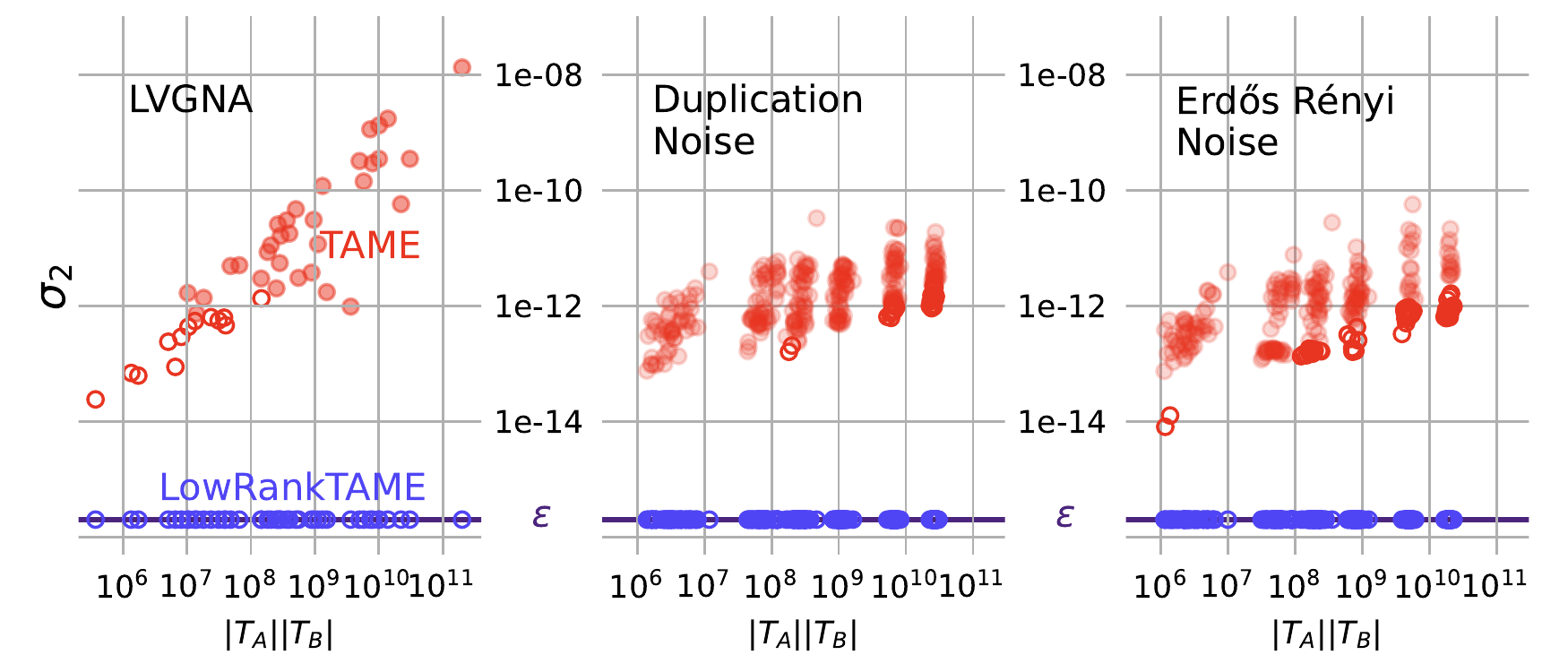}
	\caption{\label{fig:sigma2TvsLRT}
		{LowRankTAME} accurately captures the rank-structure when it is provably rank 1 for triangle adjacency tensors. TAME uses an implicit  contraction that frequently produces iterates $ \mX_\ell $, with non-dominant singular values large enough to be non-zero for matrices which are provably rank 1 ($ \alpha = 1.0 $, $ \beta =0.0 $, by lemma \eqref{lmm:rank1}.}
\end{figure}
	These experiments explain why we use the exact LowRankTAME iteration instead of the original TAME iteration to study rank when using triangle adjacency tensors. They show TAME produces iterates that would be detected as at least rank-2 even when the answer is provably rank-1, whereas LowRankTAME does not. Figure~\ref{fig:sigma2TvsLRT} plots the maximum second largest normalized singular values of $ \mX_{\ell} $, $\sigma_2 $, of all 15 iterations for TAME and LowRankTAME of the rank 1 iteration case for the LVGNA and our synthetic alignments.  Hollow points are values small enough to be considered zero (and hence, would be rank-1), and filled points are large enough to be measured as non-zero (and hence, would be rank-2). The LVGNA experiments align all pairs of distinct networks. The synthetic experiments are measured over 50 trials using random geometric graphs. Seeded networks are perturbed by both ER and Duplication noise models using default parameters.

	\subsection{PPI Graph Statistics}
	We use networks from the LVGNA project~\cite{meng2016local}, the statistics of which (unique edges and triangles) are in Table~\ref{tab:LVGNA}. 
	These networks have been aligned with a variety of contemporary methods in~\cite{vijayan2017multiple,meng2016local,nassar2018low}, to make our results comparable with prior research.  We remove any directional edges from the network before the enumerating triangles. As our methods are focused on triangle motifs, we only use networks with more than 150 triangles. We also include the largest sampled z-eigenvalue found, which -- like the standard power method -- is related to the behavior of the methods with shifts in \S\ref{sec:lowrank-tame}.


\begin{table}[!h]
\begin{minipage}{\linewidth}
\caption{LVGNA Network Statistics}
\label{tab:LVGNA}
	\begin{tabularx}{\linewidth}{ l XXXX } 
	\toprule
			Graph Name & Vertices & Edges & Triangles& Sampled $ \lambda $ \\ 
			\midrule
			worm\_Y2H1   & 2871  & 5194& 536&10.076\\
			worm\_PHY1  &  3003 & 5501& 692&12.664\\
			fly\_Y2H1        & 7094  & 23356 & 2501& 21.207\\
			yeast\_Y2H1   & 3427  &	11348& 9503&56.680\\
			human\_Y2H1 & 9996	& 39984 &15177&72.919\\ 
			human\_PHY2 & 8283	& 19697& 19190&50.872\\
			yeast\_PHY2   & 3768  &	13654& 26295&94.564\\
			fly\_PHY1        & 7885   & 36271  & 58216& 217.541\\
			yeast\_PHY1   & 6168  &	82368& 381812&454.921\\
			human\_PHY1 &16060	&157649 & 525238&488.136\\ 
		\bottomrule
	\end{tabularx}
\end{minipage}
\end{table}


	

%

\bibliographystyle{dgleich-bib3}
\bibliography{simax_TAME.bib}


\appendixpageoff
\appendixtitleoff
\renewcommand{\appendixtocname}{Supplementary material}
\begin{appendices}
	\crefalias{section}{supp}
    \setcounter{section}{19}
	\section*{Supplemental Materials}

In our supplement we include proofs to lemmas~\ref{lmm:rank1} \textit{\&}~\ref{lmm:rank-n} in~\ref{sup:proofs} and additional experiments to provide more context for our discussed results in section~\ref{sec:experiments}. Our figures include:
\begin{enumerate}
	\item a full plot of the iterate runtimes and ranks of the experiments in figure~\ref{subfig:TAME_clique_scaling_summarized},
	\item additional experiments for all the parameters in our synthetic alignment noise models put alongside our scaling experiments in figure~\ref{subfig:nSizeSynthExp}, \textit{and}
	\item the full breakdown of the runtimes and triangle match scores of the LVGNA alignment experiments, with additional results from other comparable alignment algorithms. 
\end{enumerate}

\subsection{Proofs}
\label{sup:proofs}
Here we include our own proofs of lemmas~\ref{lmm:rank1} \textit{\&}~\ref{lmm:rank-n}. The rank-1 tensor vector contraction lemma is proved inductively whereas our other proof is a more standard algebraic proof. We repeat the lemma statements for readability. 
\subsubsection{Rank 1 contraction}
\label{sec:lmm-rank1}
\begin{lemma}\label{sup:lmm-rank1}
	Given two $k$-mode,  cubical tensors $ \cmA $ and $ \cmB$ of dimension $m$ and $n$, respectively, and the $m \times n$ rank 1 matrix $\mX = \vu \vv^T$,  then for $1 \le p \le k$,
	\begin{align}
		(\cmB \otimes \cmA)\tvec(\mX)^{p} =  (\cmB \otimes \cmA)(\vv \otimes \vu)^{p} = \cmB\vv^{p} \otimes \cmA\vu^{p}.
	\end{align} 
\end{lemma}
\begin{proof}
	First note that for any symmetric tensor $\cmT$, we have in our notation $\cmT \vz^p = (\cmT \vz^{p-1}) \vz$ for any $2 \le p \le k$.  Our proof simply uses induction. The base case is $p=1$. Here, we split the multi-index $\ib$ (for $\cmA$) into its first component and tail $i_1, \jb$, and do the same for the $\ibp = (i_1', \jb')$ for $\cmB$. Then
	\[ \begin{aligned} 
		(\cmB\vu \otimes \cmA\vv)[\ileave{\jb}{\jbp}]&  = (\sum_{i_1\vphantom{'}} \cA(i_1,\jb)\vu(i_1)) (\sum_{i_1'}\cB(i_1',\jbp)\vv(i_1')) 
		= \sum_{i_1,i'_1} \cA(i_1,\jb)\cB(i_1',\jbp) \vu(i_1)\vv(i'_1) \\
		& = \sum_{\ileave{i_1}{i'_1}} (\cmB \otimes \cmA )[\ileave{i_1}{i_1'},\ileave{\jb_{}}{\jbp_{}}] (\vv \otimes \vu)[\ileave{i_1}{i'_1}] 
		= ((\cmB \otimes \cmA)(\vu \otimes \vv))[\ileave{\jb}{\jbp}]. 
	\end{aligned}	\]
	The remainder of the argument is as follows. Assume the result holds for up to some value of $p$, then we can show it holds for $p+1$ via 
	\[ \cmB \vv^{p+1} \kron \cmA \vu^{p+1} = ( \underbrace{\cmB \vv^p}_{=\cmD} \vu ) \kron (\underbrace{\cmA \vu^{p}}_{=\cmC} \vv) = (\cmD \kron \cmC)(\vv \kron \vu). \]
	Note that $\cmC$ and $\cmD$ are both symmetric, so we can apply the previous result (or the inductive hypothesis). Inductively, now, we have $\cmD \kron \cmC = (\cmB \kron \cmA) (\vv \kron \vu)^p$ and thus we are done using our initial note. 
\end{proof}
\subsubsection{Rank r contraction}
\label{sec:lmm-rank-n}
\begin{lemma}\label{sup:lmm-rankn}
	Given two $k$-mode,  cubical tensors $ \cmA $ and $ \cmB$ of dimension $m$ and $n$, respectively, and the matrix $\mX \in \RR^{m \times n}$ with the rank $ r $ decomposition $ \mY \mZ^T$,  then
	\begin{equation*} \begin{aligned}
			(\cmB \otimes \cmA)\tvec(\mX)^{p}  &=  \textstyle\sum\limits_{\mathclap{\ib = [r]^p}} \cmB(\mZ(:,i_1), \ldots, \mZ(:,i_p)) \kron \cmA(\mY(:,i_1), \ldots, \mY(:,i_p)) \\
			&= ((\cmB \modetimes{\mZ}) \kron (\cmA \modetimes{\mY}))\tvec(\mI)^p,
	\end{aligned} \end{equation*}
	where $ \mI $ is the $ r \times r $ identity matrix. 
\end{lemma}

\begin{proof}
	We show this directly on each element, starting with the first line. Let $\jb \in [m]^{k-p}$ and $\jbp \in  [n]^{k-p}$. Then
	\begin{equation*}\begin{aligned}
			((\cmB \otimes \cmA)&\tvec(\mX)^p)[\ileave{\jb}{\jbp}] 
			= \textstyle \sum_{\ileave{\lb}{\lbp}} (\cmB \kron \cmA) [\ileave{\lb}{\lbp},\ileave{\jb}{\jbp}]X(\ell_1,\ell'_1) \dots X(\ell_p,\ell'_p)  \\
			&= \textstyle \sum_{\lb,\lbp} \cA(\lb,\jb)\cB(\lbp,\jbp)
			\big[\textstyle\sum_{\ib \in [r]^p}Y(\ell_1,i_1)Z^T(i_1,\ell'_1)  \dots Y(\ell_p,i_p)Z^T(i_p,\ell'_p) \,\,\big] \\
			&= \textstyle \sum_{\ib \in [r]^p}\big[ \sum_{\lb} \cA(\lb,\jb)Y(\ell_1,i_1) \dots Y(\ell_p,i_p)\big]
			\big[\textstyle\sum_{\lbp}\cB(\lbp,\jbp)Z(\ell'_1,i_1)  \dots Z(\ell'_p,i_p) \big] \\		
			&= \textstyle \sum_{{\ib \in [r]^p}} \; [\cmA(\mY(:,i_1), \dots ,\mY(:,i_p)](\jb) [\cmB(\mZ(:,i_1)  \dots, \mZ(:,i_p))](\jbp). \\		
	\end{aligned}\end{equation*}
	This gets us the second line after we convert to the Kronecker product. 
	To get the last line, we can convert the summation over $ \ib $ into a double summation multiplied by an indicator. Essentially, we use $\sum_i a_i b_i = \sum_{ij} a_i b_j \mI_{ij}$. Note that because $\ib \in [r]^p$ we have the $r \times r$ matrix $\mI$ satisfies $ \mI(i_1,i_1) = \dots = \mI(i_p,i_p) = 1 $. Then 
	\begin{equation*}\begin{aligned}
			& ((\cmB \otimes \cmA)\tvec(\mX)^p)[\ileave{\jb}{\jbp}]\\ 
			& \qquad = \textstyle \sum\limits_{\ib \in [r]^p,  \ibp \in [r]^p}[\cmA(\mY(:,i_1), \dots ,\mY(:,i_p)](\jb) I(i_1,i_1') \dots I(i_p,i_p')[\cmB(\mZ(:,i_1')  \dots, \mZ(:,i_p'))](\jbp)     \\[1ex]
			& \qquad = \textstyle \sum\limits_{{\ileave{\ib}{\ibp} }} (\cmB \modetimes{\mZ})\!\kron\!(\cmA \modetimes{\mY})[\ileave{\ib}{\ibp}\!,\ileave{\jb}{\jbp} ] I(i_1,i'_1) \dots I(i_p,i'_p).\\[1ex]
	\end{aligned}\end{equation*}
	This final result is then just the $[\ileave{\jb}{\jbp}]$ element of $(\cmB \modetimes{\mZ}) \kron (\cmA \modetimes{\mY})\tvec(\mI)^p$.
\end{proof}		

\subsection{Further Experiments and Additional Results}
\label{sup:FutherExperiments} 
\subsubsection{Maximum Rank}
\label{sup:FExpsMaximumRank} 

\begin{figure}[h]
	\centering
	\label{subfig:TAME_clique_scaling_detailed}
	\includegraphics[width=22pc,valign=c]{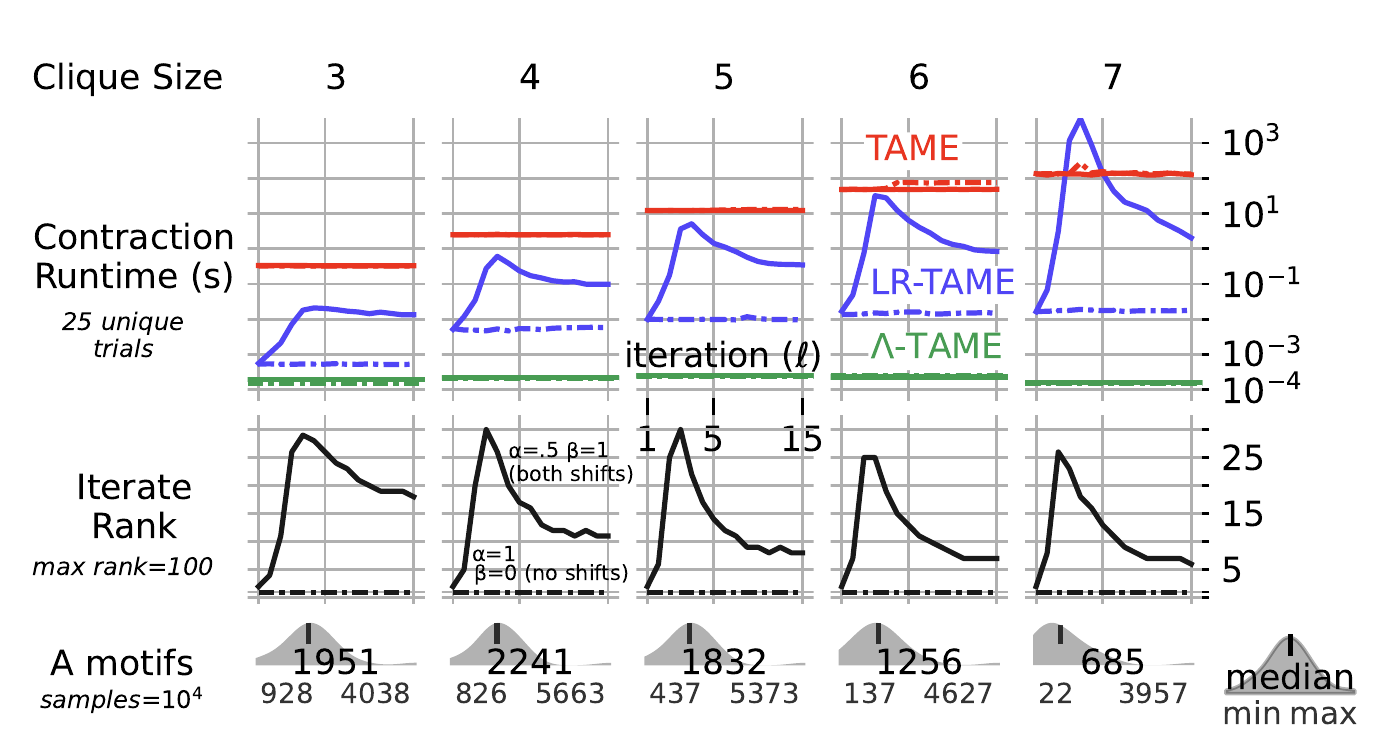}
	\caption{
		A more detailed analysis of our summarized results in figure~\ref{subfig:TAME_clique_scaling_summarized}.  We see that time spent computing contractions for TAME and LowRankTAME grows as the motif size increases, even though the rank of TAME's iterates and number of motifs declines. In spite of this, we see that the runtime of \LambdaTAME remains relatively constant across each experiment. Experiments are conducted on the smallest synthetic experiments in subfigure~\ref{subfig:maxRankExps} ($ |V_A| = 100 $) perturbed with 20\% duplicated nodes. We build our adjacency tensors using $ 10^4 $ samples of TuranShadow. We report the median of the runtimes and rank across iterations $ \ell $ over 25 trials.}
	\label{fig:FExpsrank-results}
\end{figure}

Here we show a breakdown of the runtime and ranks of each of the TAME iterates used for our results in figure~\ref{subfig:TAME_clique_scaling_summarized}. We're able to see that the rank behavior peaks at the earlier iterates of the method and slowly converge to a lower rank iterate (when using both shifts). This is notable as in the original TAME work, the authors observed that the highest quality alignments occurred within the first few iterations \citet[section 5.6]{mohammadi2017triangular}. In our experiments we didn't find that iteration rank was a better proxy for alignment accuracy than triangles align in synthetic experiments.
\subsubsection{Random Graph Models}
\label{sup:FExpsRandomGraphModels} 

In figure~\ref{subfig:nSizeSynthExp} we presented out synthetic alignment results as the reference network's size was varied, without discussion of our perturbation model's parameters. Here we include experiments where we vary the remaining modes for \ErdosRenyi and the partial duplication perturbation models. The \ErdosRenyi model from \cite{feizi2019spectral} removes (and adds) edges as a function of $p_{remove}$. The partial duplication model randomly duplicates an existing node and uniformly retains each edge with probability $p_{edge}$. We additionally may vary how long we run the procedure to control how many more vertices are introduced. Our full results are reported in figure~\ref{fig:FExpsRandomGraphAccuracy} along side the original network size experiments for easier comparison. 

\begin{figure*}[t]
	\centering
	\vspace{-5mm}
	\subfloat[\ErdosRenyi Noise]{\includegraphics[width=.4\textwidth,valign=c]{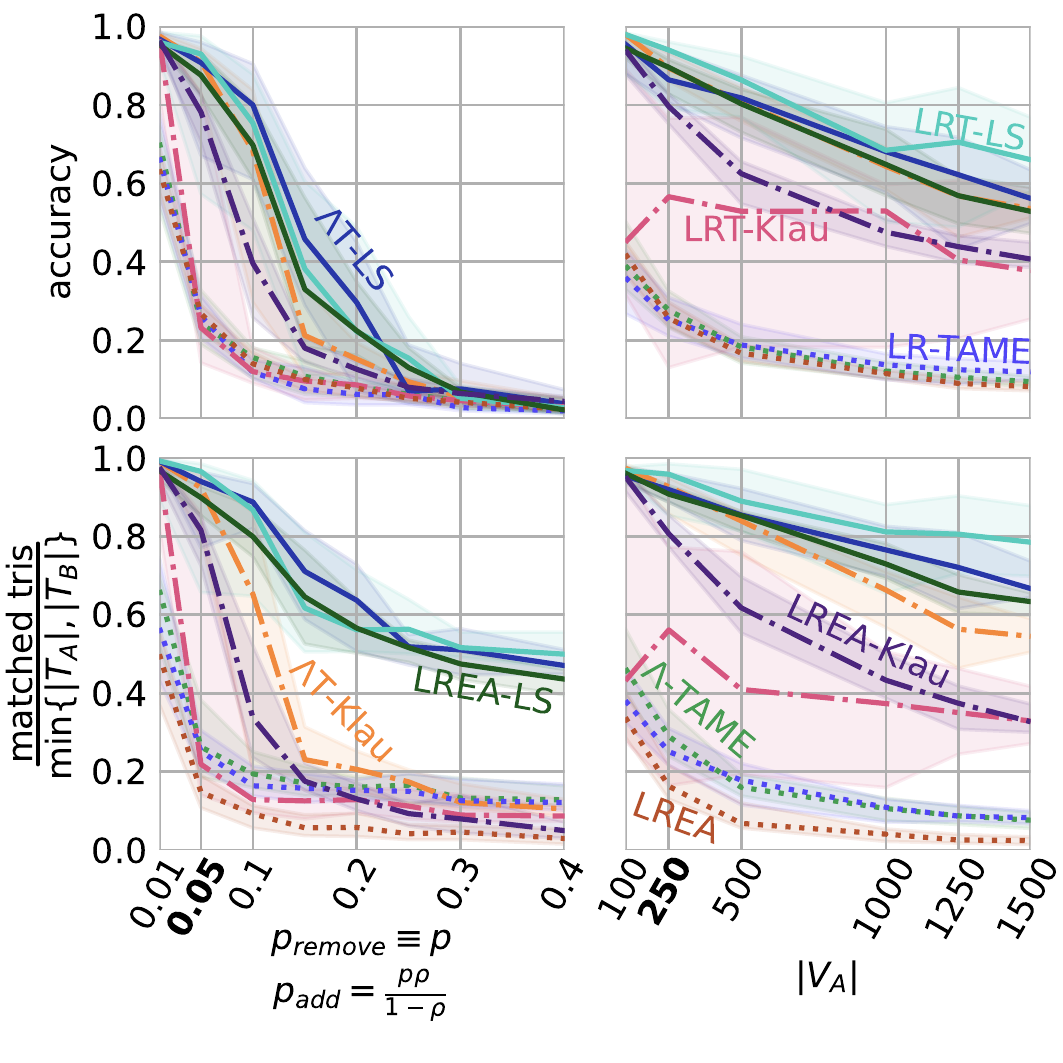}
		\label{ERNoise}}
	\subfloat[Duplication Noise]{\includegraphics[width=.575\textwidth,valign=c]{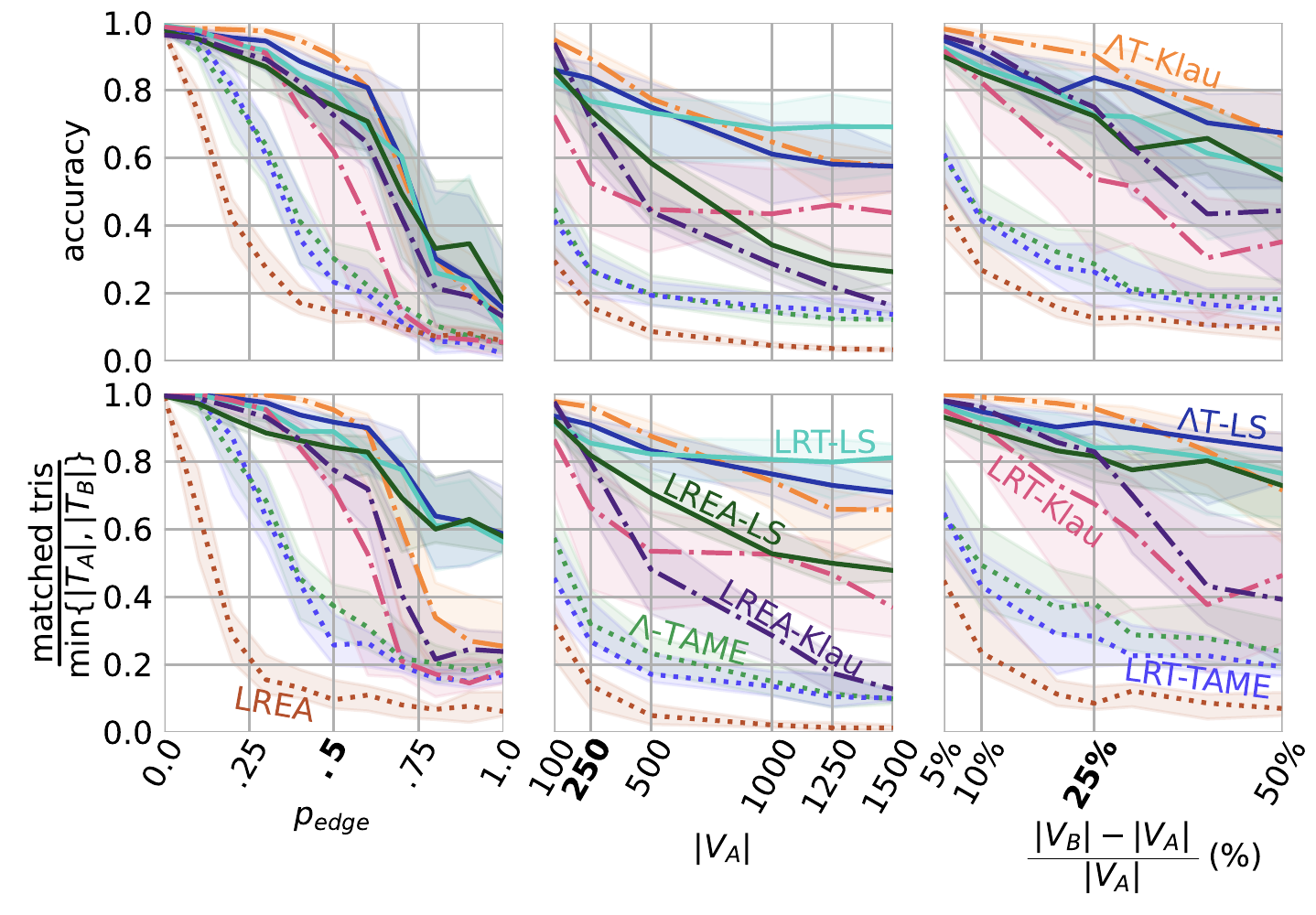}
		\label{DuplicationNoise}}
	\caption{\label{fig:FExpsRandomGraphAccuracy}Refining \LambdaTAME and LowRankTAME's embeddings with Klau's Algorithm and local search provides accurate matchings using \ErdosRenyi and duplication noise models. Local search using either LowRankTAME or \LambdaTAME's embeddings are the most accurate for ER noise, while \LambdaTAME's embeddings improved with Klau's algorithm becomes competitive for duplication noise. LowRankTAME using local search gives the best performance on the scaling experiments, but only slightly outdoes \LambdaTAME for a longer runtime as was seen in figure~\ref{subfig:LVGNA_postProcessing}. Triangles matched is a good proxy for accuracy for lower noise levels. We plot the median of 20 trials with 20th-80th percentile ribbons. Default parameters for other experiments are bolded along the x-axis.} 
\end{figure*}

Overall we see similar performance to the size scaling experiments for the other modes, though refined LowRankEigenAlign does better for the additional parameters in the duplication noise models. A notable difference from just the size scaled experiments is that for the $ p_{remove} $ and $ p_{edge} $ parameters of each noise model, the triangle match rate is much higher than the accuracy. Though the relative performance between the methods is still maintained, this seems to detract a little from it's role as proxy for accuracy. However this is reasonable given that for the inflation only occurs for the largest amounts of noise used. For the duplication noise, the largest values of $p_{edge} $ is 100\% which means that all new duplicated nodes are isomorphic to their original nodes. When permuted, it's unreasonable to expect an unsupervised algorithm to be able to determine between the original and duplicated node. An analogous problem will occur for the largest values of $ p_{remove}$. Removing and adding in 40\% of the potential edges in a network is an extreme perturbation. 

\subsubsection{LVGNA Alignments}
\label{sup:LVGNA_breakdown}
\begin{figure*}[t]
	\centering
	\includegraphics[scale=.65]{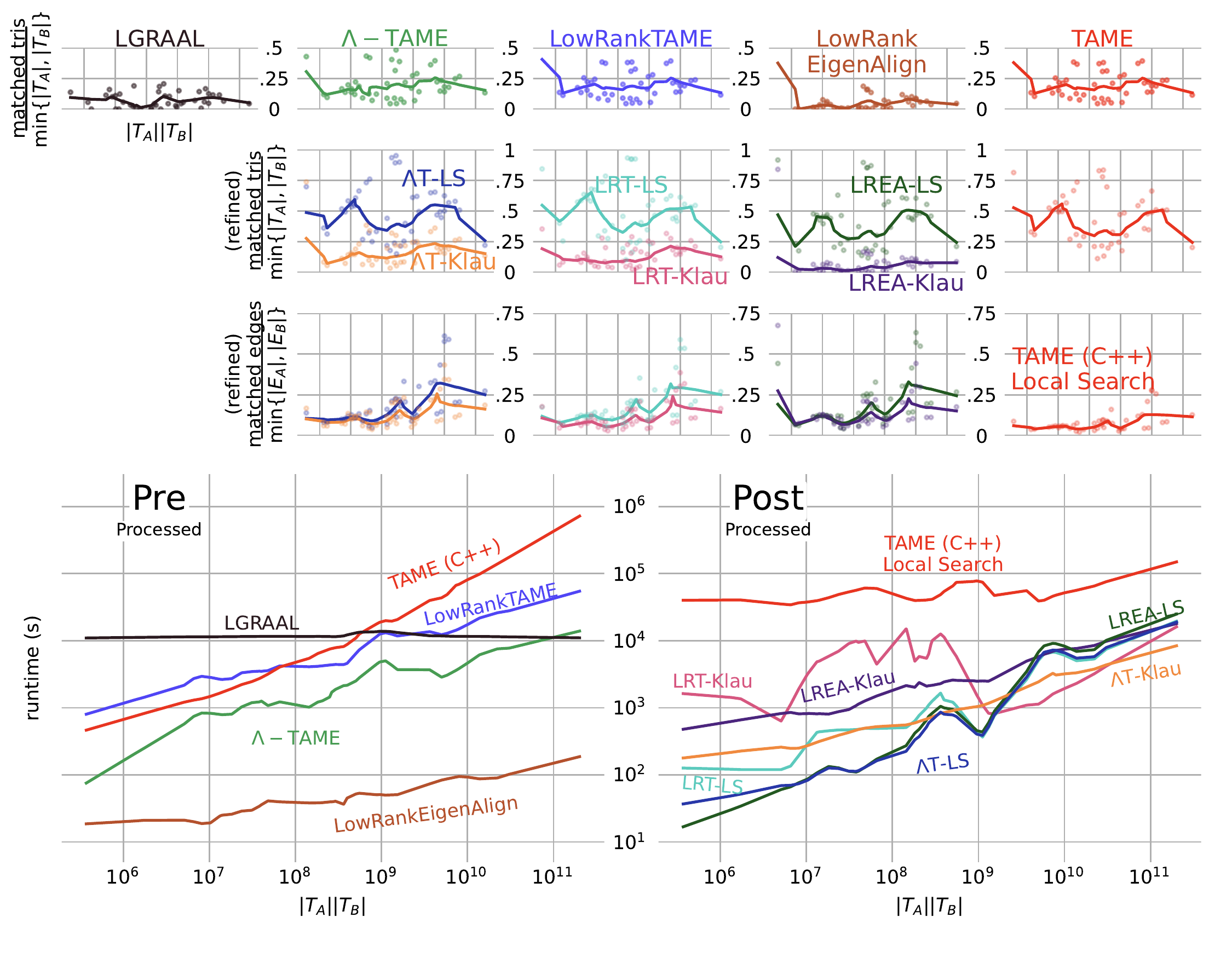}
	\vspace{-5mm}
	\caption{\label{supfig:LVGNA_postProcessing}We report the full performance of each tested alignment method on the LVGNA network(reported in~\ref{subfig:LVGNA_postProcessing}). Refined results align a comparable number of triangles and edges to the original TAME method while running orders of magnitudes faster when tested on the LVGNA network. All refining methods run faster than TAME's b-match local search implementation. We include the preprocessed Triangle matching rates at the top, and include the refined results below. Triangle and edge matching rates are separated by pre processing algorithm for readability. All plots utilize loess smoothing (using approximate 30\% neighbors).  
	}
\end{figure*}
Here we include the full results of the alignment algorithms tested on the LVGNA networks, instead of the relative performance of each algorithm compared to TAME (reported in figure~\ref{subfig:LVGNA_postProcessing}).
These results show that LowRankTAME is roughly an order of magnitude faster than the implementation of TAME in C++ from~\cite{mohammadi2017triangular} for the largest experiments. \LambdaTAME is about two orders of magnitude faster. Given what we see in the bottlenecks demonstrated in figure~\ref{subfig:LVGNA_TTVMatchingRatios} we can attribute the growth in runtime to growing cost of solving for the maximum weighted matchings instead of the tensor operations. In addition to the previously discussed methods, we also include results using LGRAAL~\cite{malod2015graal}. L-GRAAL computes graphlet degrees to guide alignments (similar to other GRAAL methods) and uses Lagrangian relaxations to seed and extend its matchings. L-GRAAL was one of the higher quality, but longer running algorithms in the original TAME paper~\cite[Figure 1,2]{mohammadi2017triangular}. We did not find that LGRAAL aligned a sufficient amount of triangles (this is likely due to the objective function focusing on aligning edges). Additionally LGRAAL doesn't provide an output that we can refine like we can for LowRankEigenAlign or the TAME methods. We include it's runtime for an additional comparison point to our methods. Its relatively constant runtime is due to the method consistently hitting it's maximum runtime limit before converging to a solution.
We see that LowRankEigenAlign is the fastest method, and when refined provides more edges aligned, but the triangles aligned are not as competitive as the refined LowRankTAME or \LambdaTAME embeddings (as was seen in figure~\ref{subfig:LVGNA_postProcessing}). We also see that with the exception of the largest alignment problems, the refining runtime is eclipsed by the runtime of the TAME methods. LowRankEigenAlign was the only method to be faster than the refinement. It should also be noted that though the LowRankTAME and TAME triangle matching rates look identical, there are small differences which amount due how ties are broken in the rounding procedures of LowRankTAME and TAME. Edges in TAME's C++ primal dual implementation are traversed using a row major formatting, whereas our julia implementation uses a column major formatting.

\end{appendices}


\end{document}